\documentclass[twocolumn]{revtex4-2}
\usepackage{graphicx}
\usepackage{amsmath, amssymb, amsthm}
\usepackage{todonotes}
\usepackage{physics}
\usepackage{subcaption}
\usepackage{float}

\newcommand{\haf}{\operatorname{Haf}}

\newtheorem{prop}{Proposition}
\newtheorem{thm}{Theorem}

\begin{document}

\begin{abstract}
    Gaussian Boson Samplers aim to demonstrate quantum advantage by performing a sampling task believed to be classically hard. The probabilities of individual outcomes in the sampling experiment are determined by the Hafnian of an appropriately constructed symmetric matrix. For nonnegative matrices, there is a family of randomized estimators of the Hafnian based on generating a particular random matrix and calculating its determinant. While these estimators are unbiased (the mean of the determinant is equal to the Hafnian of interest), their variance may be so high as to prevent an efficient estimation. Here we investigate the performance of two such estimators, which we call the Barvinok and Godsil-Gutman estimators. We find that in general both estimators perform well for adjacency matrices of random graphs, demonstrating a slow growth of variance with the size of the problem. Nonetheless, there are simple examples where both estimators show high variance, requiring an exponential number of samples. In addition, we calculate the asymptotic behavior of the variance for the complete graph. Finally, we simulate the Gaussian Boson Sampling using the Godsil-Gutman estimator and show that this technique can successfully reproduce low-order correlation functions.

\end{abstract}

\title{On randomized estimators of the Hafnian of a nonnegative matrix}


\author{Alexey Uvarov}%
 \email{alexejuv@gmail.com}
\author{Dmitry Vinichenko}
\affiliation{Skolkovo Institute of Science and Technology, Bolshoy Boulevard, 30, p.1, Moscow 121205, Russia}

\date{\today}

\maketitle

\section{Introduction}

Counting the number of perfect matchings in a graph is a well-known computationally hard problem. The function that takes an adjacency matrix and counts the perfect matchings (possibly with weights), known as the Hafnian, appears in consideration of a model of quantum computation called Gaussian Boson Sampling (GBS)~\cite{hamilton_gaussian_2017,kruse_detailed_2019}. In this model, a particular optical setup is used to produce samples from a probability distribution which is believed to be hard to reproduce classically. The hardness stems from the fact that calculating the probability of a particular measurement outcome requires the calculation of the abovementioned Hafnian of a matrix that depends on the outcome.


Recently, a series of GBS experiments~\cite{zhong_quantum_2020,zhong_phase-programmable_2021,deng_solving_2023,deng_gaussian_2023} demonstrated improved quantum sampling capabilities, however, it was later shown that these devices can be simulated classically~\cite{oh_classical_2023-1}. Still, even when the technology does reach the quantum advantage milestone, there will be a long way from an efficient non-classical sampler to a device that would assist in solving practical problems.

A proposed route to applications consists in using the sampler as an approximate solver for the densest subgraph problem and related problems~\cite{arrazola_using_2018,arrazola_quantum_2018,bromley_applications_2020,banchi_molecular_2020}. The kernel matrix of the Gaussian state is taken to be a rescaled adjacency matrix of the problem graph. That way, the detection events correspond to vertex subsets and their induced subgraphs. The probability of observing a particular event is proportional to the Hafnian of the induced subgraph, which is correlated with its density. Thus, the most likely samples are approximate solutions of the densest subgraph problem, or at least good starting points for a classical search algorithm. 
Crucially, in all graph-related applications to date~\cite{arrazola_using_2018,banchi_molecular_2020, anteneh_sample_2023,schuld_measuring_2020}, the kernel matrix is nonnegative, meaning that all its elements are greater or equal than zero. 

This special case has recently been considered as potentially accessible for classical simulation. In Ref.~\cite{quesada_exact_2020}, Quesada and Arrazola proposed a method for classical simulation of general GBS which relies on the calculation of conditional probabilities. For nonnegative kernels, the authors propose a classical simulation algorithm that does not necessarily use exponential time. The key element of this algorithm is a probabilistic method of estimating the Hafnian of a nonnegative matrix due to Barvinok~\cite{barvinok_polynomial_1999}.

The Barvinok estimator takes the input matrix and returns a random variable with mean $\mu$ and variance $\sigma^2$, both of which depend on the input. The mean is guaranteed to be equal to the Hafnian, but to achieve good relative accuracy, we need to run the algorithm multiple times and take the average of the outputs. The size of the sample has to scale as $\sigma^2 / \mu^2$, therefore the analysis of the scaling of the relative variance $\sigma / \mu$ is crucial in understanding the performance of the algorithm.

In this work, we investigate the behavior of the Barvinok estimator and a related method known as the Godsil-Gutman estimator~\cite{godsil1978matching}. Our main findings are the following:
\begin{enumerate}
    \item We provide an analytical formula for $\sigma / \mu$ as a sum over certain graph coverings. We show that the variance of the Godsil-Gutman estimator is always smaller or equal to that of the Barvinok estimator. We also show that, for certain classes of matrices, the relative variance $\sigma / \mu$ scales exponentially with the size of the problem. In addition, we provide a new asymptotic estimate of $\sigma$ for the case when the input matrix is the adjacency matrix of the complete graph.
    \item We show numerically that for adjacency matrices of random graphs sampled from the Erd\"os-R\'enyi ensemble, both estimators demonstrate a modest growth of variance across the whole range of graph densities. Specifically, for the complete graph, $\sigma / \mu$ grows as a square root of the graph size, and for all other densities, the average $\sigma / \mu$ is smaller than that.
\end{enumerate}

The paper is structured as follows. In Section~\ref{sec:estimator_variances}, we derive an expression for the variance of the Barvinok and Godsil-Gutman estimators in terms of so-called perfect 2-matchings. We also provide examples of graphs that do not admit a polynomial approximation by these estimators without additional effort (Sections \ref{sec:union_edges} and~\ref{sec:union_cycles}). We present our findings on the variance for the complete graph in Section~\ref{sec:complete_graphs}.
In Section~\ref{sec:numerics_estimators}, we show the results of numerical investigations for random graphs. Finally, in Section~\ref{sec:numerics_gbs}, we report on a classical simulation of Gaussian Boson Sampling with nonnegative kernels. Section~\ref{sec:conclude} contains concluding remarks.


\section{Randomized estimators of the Hafnian}
\label{sec:estimator_variances}

The Hafnian of a symmetric, even-dimensional matrix $A \in \mathbb{C}^{2m \times 2m}$ is defined as follows:
\begin{equation}
    \haf A = \sum_{\pi \in PM (2m)} \prod_{\{i, j\} \in \pi} a_{ij}.
    \label{eq:haf_def}
\end{equation}
Here $\pi \in PM(2m)$ means that the summation goes over perfect matchings of $2m$ objects. A perfect matching is simply a way to split $2m$ objects into pairs. For example, if $m=2$, there are three perfect matchings: $\{\{1, 2\}, \{3, 4\}\}$, $\{\{1, 3\}, \{2, 4\}\}$, $\{\{1, 4\}, \{2, 3\}\}$.

We will often treat $A$ as an adjacency matrix of a weighted graph with vertex set $V = \{1, ..., 2m\}$ and edge set $E \subseteq \{\{i,j\}|i, j \in V,  i < j\}$. 
For graphs, perfect matchings are all possible ways to split all vertices into pairs so that each pair is connected by an edge; if there is no edge between vertices $i$ and $j$, then $a_{ij}=0$, and the matchings including the pair $\{i, j\}$ do not contribute to the sum in~\eqref{eq:haf_def}. A schematic depiction is shown in~Fig.~\ref{fig:hafnian_example}.

\begin{figure}
    \centering
    \includegraphics[width=\linewidth]{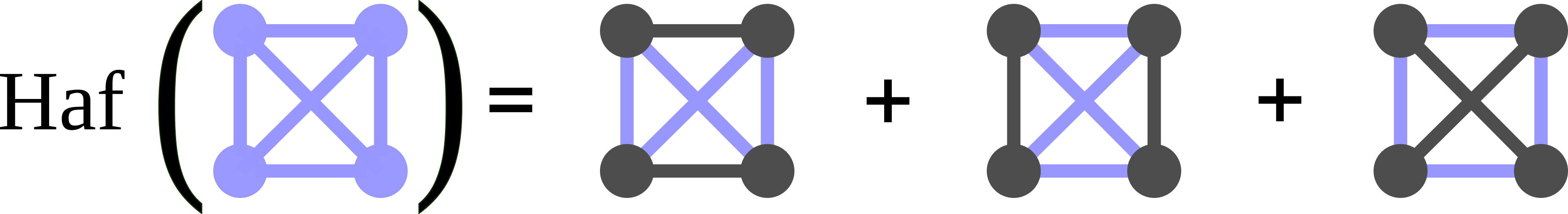}
    \caption{Perfect matchings counted by the Hafnian function.}
    \label{fig:hafnian_example}
\end{figure}




The Hafnian of a nonnegative matrix admits the following randomized approximation algorithm. Let $W$ be a skew-symmetric matrix whose above-diagonal entries are i.i.d.~random variables such that $\mathbb{E} w_{ij} = 0$ and $\mathbb{E} w_{ij}^2 = 1$. 
Define $G$ to be a matrix such that its entries are equal to $g_{ij} = w_{ij} \sqrt{a_{ij}}$. Then one can show that 
\begin{equation}
\label{eq:expected_det_is_haf}
    \mu \equiv \mathbb{E} \det G = \haf A,
\end{equation}
where the expectation is taken over the random matrix entries $w_{ij}$.
This is proven by expanding the determinant by definition and observing that the only terms which do not vanish upon taking the expected value are those where $g_{ij}$ only appear in the second power. Such terms precisely correspond to perfect matchings of the graph vertices.

To estimate $\haf A$, we then simply repeat the sampling $N$ times, calculate the determinants, and take the sample mean. When $w_{ij}$ are sampled from the discrete uniform distribution on $\{-1, 1\}$, we call this protocol the Godsil-Gutman estimator~\cite{godsil1978matching}, when $w_{ij}$ are sampled from the standard normal distribution, we call it the Barvinok estimator~\cite{barvinok_polynomial_1999}. An important difference occurs in the fourth moment: for the standard normal distribution, $\mathbb{E} w_{ij}^4 = 3$, while for the discrete uniform distribution, $\mathbb{E} w_{ij}^4 = 1$.

The accuracy of the estimator depends on the variance $\sigma^2 = \operatorname{Var} \det G$. If we want to achieve multiplicative error $\epsilon$, then the standard error of the mean $\sigma / \sqrt{N}$ has to be of the order $\mu \epsilon$. In particular, if the relative standard deviation $\sigma / \mu$ scales exponentially with the size of the problem, then $N$ will also have to scale exponentially.


In the following we derive an expression for $\mathbb{E} \det G^2$ and for $\left(\mathbb{E} \det G \right)^2$. For this, we will need the notion of a perfect 2-matching, which generalizes the idea of a matching. A perfect 2-matching of a graph is a spanning subgraph (i.e.~it contains all the vertices), such that every connected component is an edge or a cycle. In particular, we will consider 2-matchings that only contain cycles of even length. In simpler terms, a perfect matching is a way to cover all vertices with edges, touching every vertex exactly once, while a perfect 2-matching is a way to cover all vertices with edges and cycles, touching every vertex exactly once. Examples of perfect 2-matchings are shown in~Fig.~\ref{fig:2-matchings}.

\begin{figure}
    \centering
    \includegraphics[width=\linewidth]{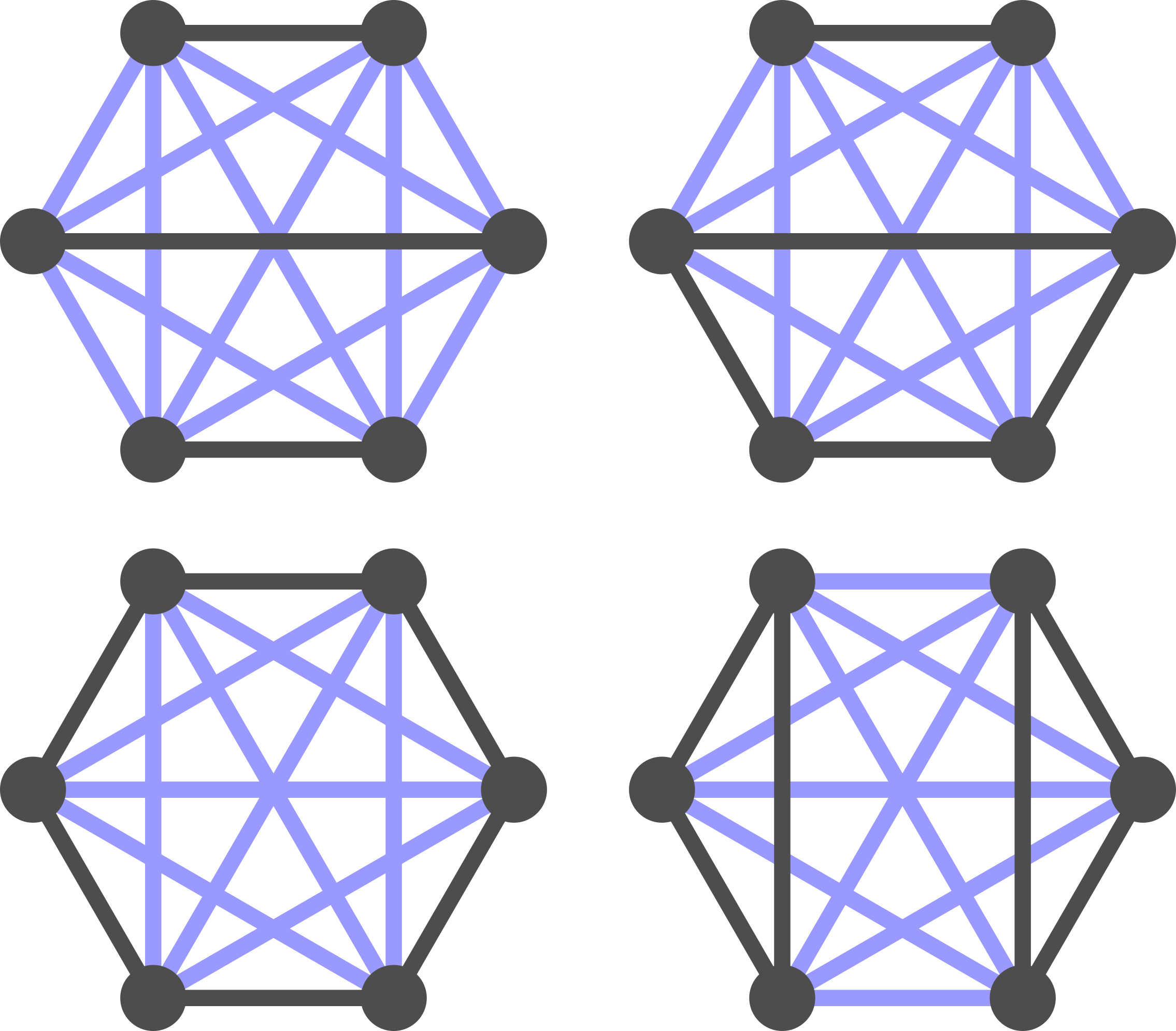}
    \caption{Examples of perfect 2-matchings. The one on the bottom right contains cycles of odd length; such 2-matchings do not contribute to the sum~\eqref{eq:expected_det2}.}
    \label{fig:2-matchings}
\end{figure}

In the following we derive an expression for $\mathbb{E} \det G^2$ and for $\left(\mathbb{E} \det G \right)^2$ in terms of perfect 2-matchings with even cycles. We will say that a permutation $\pi \in S_{2m}$ matches vertices $i$ and $j$ if $\pi(i) = j$ and $\pi(j) = i$.

\begin{prop}
\label{prop:second_moment}
Let $\mathbb{E} w_{ij}^3 = 0, \mathbb{E} w_{ij}^4 = \eta$. Then 
    \begin{equation}
        \label{eq:expected_det2}
        \mathbb{E} (\det G)^2 = \sum_d \eta^{|\mathrm{match}(d)|}6^{|\mathrm{cycle}(d)|}
        \prod_{\substack{\{i,j\} \in \mathrm{match}(d) \\ \{k,l\} \in d \setminus \mathrm{match}(d)}} a_{ij}^2 a_{kl}.
    \end{equation}
    Here the sum is taken over all perfect 2-macthings $d$ that contain cycles of even length; $\mathrm{match}(d)$ is the set of  isolated edges in $d$; $\mathrm{cycle}(d)$ is the set of even-length cycles in $d$.
\end{prop}


\begin{proof}
    The squared determinant is the product of two sums over permutations:
    \begin{multline}
          (\det G)^2 = \left( \sum_{\pi \in S_{2m}} \operatorname{sgn} \pi \prod_{i=1}^{2m} g_{i \pi(i)}\right) \times \\
        \times \left( \sum_{\pi' \in S_{2m}} \operatorname{sgn} \pi' \prod_{i'=1}^{2m} g_{i' \pi'(i')}\right).   
        \label{eq:det2}
    \end{multline}
    The terms that do not vanish under taking the expected value are monomials in $g_{ij}$ in which all terms appear either in the second or in the fourth power. An element $g_{ij}$ appears four times if both permutations match vertices $i$ and $j$ (e.g.~Fig.~\ref{fig:2-matchings}, top left). In this situation, taking the expected value over this particular multiple yields a factor $\eta a_{ij}^2$ (the minus signs coming from the skew-symmetry of $G$ are eliminated here). A variable $g_{ij}$ can appear in the second power in the following cases:
    \begin{enumerate}
        \item Exactly one of the permutations matches vertices $i$ and $j$, and the other permutation matches vertices $i$ and $j$ to some other vertices (e.g.~Fig.~\ref{fig:2-matchings}, right);
        \item The vertex $i$ has the same orbit of size $\geq 3$ under the actions of both $\pi$ and $\pi'$, and either $\pi(i) = j$ or $\pi(j) = i$ (e.g.~Fig.~\ref{fig:2-matchings}, bottom left);
    \end{enumerate}
    Indeed, suppose that $\pi$ matches $i, j$ but $\pi'$ does not. Then if $\pi'(i) = k$ but $\pi'(k) \neq i$, then $g_{ik}$ necessarily comes in the first power and thus vanishes. If $\pi'(i) = j$ but $\pi'(j) \neq i$, then the total power of $g_{ij}$ is three, which vanishes under taking the expected value. Therefore $\pi'$ also has to match $i$ to some vertex, and the same applies to vertex $j$.
    
    Suppose now that neither permutation matches $i$ with $j$, but $\pi(i) = j$ and $\pi'(i) = j$. Then if $\pi(j) = k \neq \pi'(j)$, the term $g_{jk}$ appears in the first power and vanishes under the expected value. Therefore, $\pi^2(i) = (\pi')^2(i)$. Continuing by induction, we find that the whole orbit of $i$ has to be same under both permutations. The same can be shown if $\pi'^{-1}(i) = j$, \textit{mutatis mutandis}.

    Let us now assign to each nonvanishing term a perfect 2-matching $d$, by adding an edge $i, j$ if either of the vertices is mapped to the other by either $\pi$ or $\pi'$. Taken together with the adjacent vertices as a subgraph, this edge subset contains only isolated edges and cycles. If a pair of permutations $\pi, \pi'$ contains a shared cycle of odd length, then we can construct permutations $\tilde{\pi}$ and $\tilde{\pi}'$ by reversing the direction of this cycle. Thus, we will have four pairs of permutations corresponding to the same subgraph. However, reversing the direction of the cycles add a factor of $(-1)$ to its contribution, which means that these terms will cancel each other. Therefore, the only relevant subgraphs consist of matchings and even cycles. Examples of how two permutations can form a perfect 2-matching are shown in~Fig.~\ref{fig:permutation_merge}.

    \begin{figure}
        \centering
        \includegraphics[width=\linewidth]{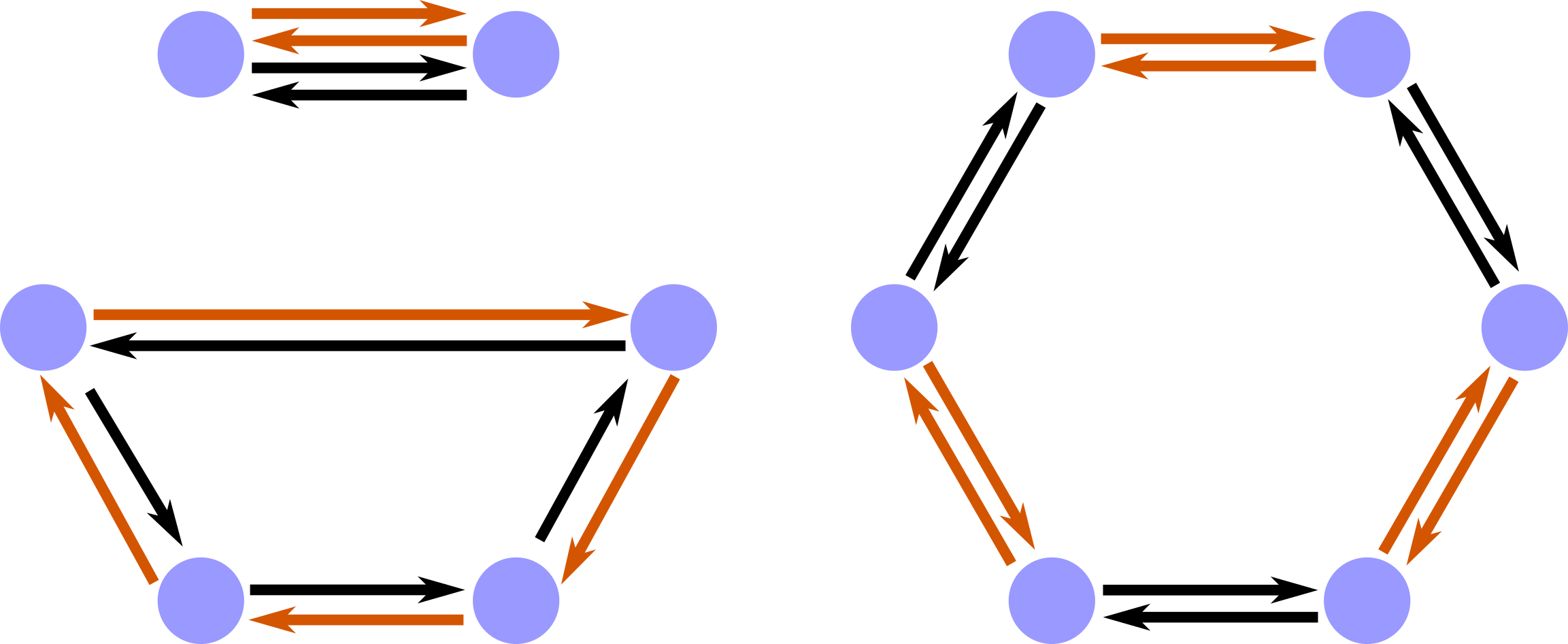}
        \caption{Examples of perfect 2-matchings formed by two permutations (orange arrows and black arrows). Note that if both permutations map two vertices to each other, they become connected by four arrows, meaning that the corresponding term in~\eqref{eq:det2} has a fourth power in $g_{ij}$.}
        \label{fig:permutation_merge}
    \end{figure}

    Suppose a cycle in $d$ traverses vertices $v_1, ..., v_{2k}$. There are only six ways two permutations can act on these vertices to form a cycle in $d$: they can match vertices in an alternating fashion in two different ways, or both permutations can cyclically shift the vertices $(v_1, ..., v_{2k})$ in either direction. Thus, each relevant subset $d$ appears in the sum exactly $6^{|\mathrm{cycle}(d)|}$ times. The cycle type of both permutations is the same, therefore $\operatorname{sgn} \pi = \operatorname{sgn} \pi'$. Thus, collecting the terms and prefactors coming from the expectations of the fourth powers of $g_{ij}$, we arrive to the desired formula.
\end{proof}

The squared Hafnian can also be expanded in similar terms:

\begin{prop}
In the definitions of Proposition \ref{prop:second_moment}, the following is true:
    \begin{equation}
        \label{eq:haf2_is_sum}
        \haf^2 A =  \sum_d 2^{|\mathrm{cycle}(d)|}
        \prod_{\substack{(i,j) \in \mathrm{match}(d) \\ (k,l) \in d \setminus \mathrm{match}(d)}} a_{ij}^2 a_{kl}.
    \end{equation}
\end{prop}

\begin{proof}
    Let us expand the squared Hafnian by definition:
    \begin{multline}
    \label{eq:haf2}
        \haf^2 A = \left(\sum_{\pi \in PM (2m)} \prod_{(i, j) \in \pi} a_{ij} \right) \times \\
        \times \left(\sum_{\pi' \in PM (2m)} \prod_{(k, l) \in \pi'} a_{kl} \right).
    \end{multline}
    The union of edges of two perfect matchings consists entirely of isolated edges and alternating cycles (the cycles are formed in a similar fashion to Fig.~\ref{fig:permutation_merge}, right). Each alternating cycle can be formed in two ways. Thus, every edge subset $d$ corresponds to $2^{|\mathrm{cycle}(d)|}$ terms in the sum (\ref{eq:haf2}).
\end{proof}

In Appendix~\ref{sec:det2_example} we show an example of the calculation of $\mathbb{E} (\det G)^2$ and $(\haf A)^2 $ for $m=2$.

In the following subsections, we will consider some special classes of matrices where the estimators show interesting behavior.

\subsection{Union of edges}
\label{sec:union_edges}


The first example is the following $2m \times 2m$ matrix:
\begin{equation}
    A = \begin{pmatrix}
    0 & 1  \\
    1 & 0 \\
    &  & 0 & 1 \\
    &  & 1 & 0 \\
    &  &  &  & 0 & 1 \\ 
    &  &  &  & 1 & 0 \\ 
    &  &  &  &  &  & \ddots &  \\ 
    &  &  &  &  &  &  & 0 & 1 \\ 
    &  &  &  &  &  &  & 1 & 0 \\ 
    \end{pmatrix}
\end{equation}
Clearly $\haf A = 1$, since this is the adjacency matrix of a graph which has exactly one perfect matching. The determinant of $G$ is equal to
\begin{equation}
    \det G = g_{12}^2 g_{34}^2 \dots g_{2m-1, 2m}^2.
\end{equation}
The expected square of the determinant is
\begin{equation}
    \mathbb{E} (\det G)^2 = \mathbb{E} g_{12}^4 g_{34}^4 \dots g_{2m-1, 2m}^4 = \eta^m.
\end{equation}
This means that the variance of the Barvinok determinant for $A$ is equal to $3^m - 1$. The variance of the Godsil-Gutman determinant, on the other hand, is equal to zero.

\subsection{Union of cycles}
\label{sec:union_cycles}
The second example is the adjacency matrix of a graph consisting of a disjoint union of $k$ cycles of length four:
\begin{equation}
    \includegraphics[width=\linewidth]{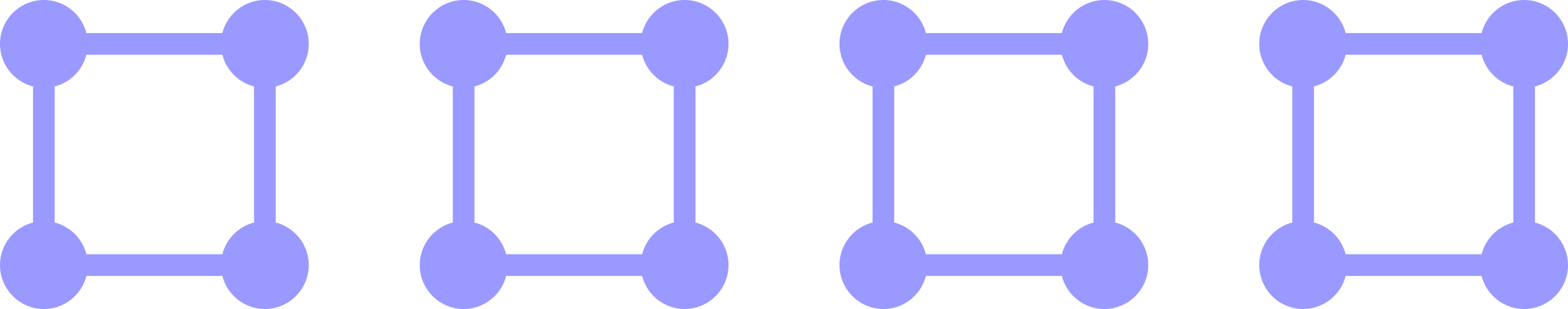} \nonumber
\end{equation}
Using (\ref{eq:expected_det2}), we quickly find that the variance of either estimator grows exponentially faster in $k$ than the squared mean. 

In general, if $\mathcal{G}$ is a connected graph whose adjacency matrix is such that $\sigma / \mu = r > 1$, then the disjoint union of $k$ copies of this graph will yield a relative deviation equal to $r^k$. 

In principle, we could slightly modify the estimation technique by considering each connected component separately, since the Hafnian of the disjoint union is the product of individual Hafnians. However, we can slightly modify the previous example to get a connected graph that whose Hafnian cannot be efficiently calculated by the Godsil-Gutman estimator:
\begin{equation}
    \includegraphics[width=\linewidth]{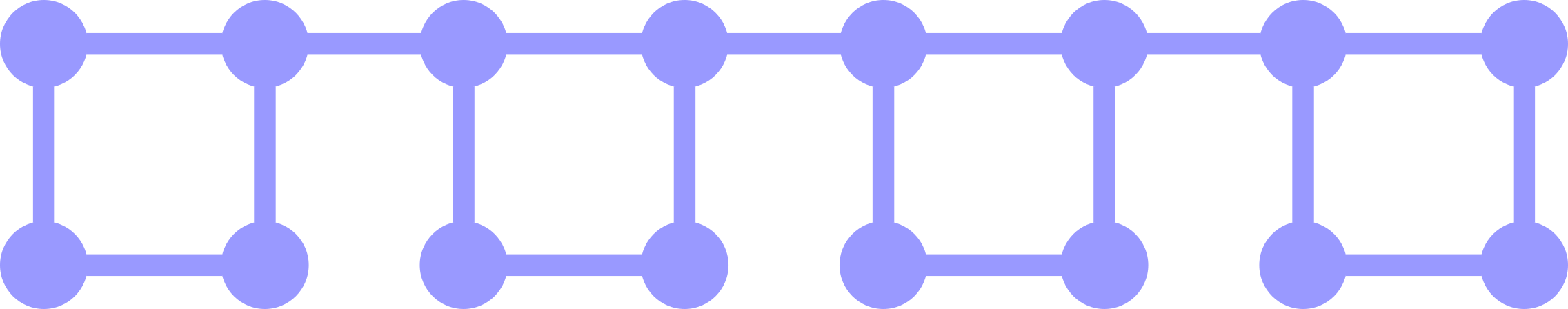} \nonumber
\end{equation}
One can verify that in this example, no new perfect matchings or perfect 2-matchings are created by adding the new edges. We could slightly modify the procedure even further by first detecting and cutting the bridges (i.e.~graph edges such that their removal makes the graph disconnected), and then considering each component separately. Nonetheless, while it is easy to adjust the algorithm to account for these simple examples, there is yet no way to guarantee that no other pathological examples will appear. 



The complex variants of the Hafnian estimator seem to suffer from all the same problems as their real counterparts. We can construct such estimators by taking $W$ to be a skew-Hermitian matrix and $g_{ij}$ to be complex random variables such that $\mathbb{E} w_{ij} = 0$ and $\mathbb{E} |w_{ij}|^2 = 1$. This can be e.g.~a Gaussian random variable or a discrete distribution over the roots of unity, as used in Ref.~\cite{karmarkar_monte-carlo_1993} for estimating the permanent. The same example as above will produce an estimate with an exponentially high relative deviation.

Finally, we note that for a random real variable with zero mean and unit variance, the fourth moment cannot be less that one. Therefore, one cannot improve on the Godsil-Gutman estimator in terms of variance simply by changing the distribution of $w_{ij}$.

\subsection{Complete graphs}
\label{sec:complete_graphs}

One more class of matrices worth discussing is the class of matrices corresponding to complete graphs. It was previously known that the relative deviation $\sigma / \mu$ of the Godsil-Gutman determinant is polynomial in size for complete graphs~\cite{lovasz_matching_2009}. Here we go further and show that for a complete graph on $2m$ vertices, $\sigma / \mu$ scales as $\sqrt{m}$. 
To establish this, we use the theory of combinatorial species~\cite{joyal_theorie_1981, bergeron_combinatorial_1997}. This theory lets us construct a function whose $m$'th Taylor coefficient contains information about the expected value $\mathbb{E} \det G^2$ for the adjacency matrix of a complete graph with $2m$ vertices. Such functions encoding numerical sequences in their Taylor coefficients are known as generating functions and are widely used in combinatorics~\cite{wilf2005generatingfunctionology}. Then we use the so-called Darboux's method~\cite{wilf2005generatingfunctionology,flajolet2009analytic} to calculate the asymptotic behavior of the Taylor coefficients. Here we will only state the final result, leaving the proof to Appendix~\ref{sec:species}.

\begin{thm}
\label{thm:complete_graph}
    Let $A$ be the adjacency matrix of a complete graph with $2m$ vertices. Then we have the following as~$m \rightarrow \infty$:
    \begin{equation}
       \frac{\mathbb{E} \det G^2}{(\haf A)^2}  = \sqrt{\pi} m e^{\frac{\eta - 3}{2}} + O(1).     
    \end{equation}
\end{thm}

\section{Relative error for random graphs}
\label{sec:numerics_estimators}

In this section, we investigate the behavior of Godsil-Gutman and Barvinok estimators for adjacency matrices of random graphs (that is, $a_{ij} \in \{0, 1\}$). 
We generate random graphs from the the Erd\"os-R\'enyi ensemble $G_{2m, p}$. That is, we take $2m$ vertices and randomly connect them, adding each edge independently with probability $p$.
For each graph, we take $N = 500 \cdot m$ samples of the estimator and compute their sample mean and sample variance. We also compute the exact Hafnian of the graph. Using this data, we evaluate the relative standard deviation $\sigma / \mu$. Here, the choice of sample size $N$ is informed by the results of the previous section, implying that under this scaling of sample size, the complete graph should demonstrate an asymptotically constant relative error. Finally, we estimate the accuracy of determining $\sigma / \mu$ by bootstrap resampling ($n = 100$).

Fig.~\ref{fig:relsigma} shows the behavior of $\sigma / \mu$ for both estimators. The average $\sigma / \mu$ for the Godsil-Gutman estimator (Fig.~\ref{fig:relsigma}a) exhibits a monotonous growth with $p$ for all values of $m$ considered. Interestingly, the error first grows with $p$ quite sharply, then flattens out and becomes relatively constant. The width of the growth period and its location both change with the increase of system size, shrinking and moving towards smaller values of $p$. The average error for the Barvinok estimator (Fig.~\ref{fig:relsigma}b) demonstrates a substantially more complicated behavior. Instead of being monotonous with the density, the error shows a distinct peak. The location of the peak appears to match the location of the sharp growth region for the Godsil-Gutman estimator. 

\begin{figure*}
    \centering
    \subfloat[]{
        \includegraphics[width=0.5\textwidth]{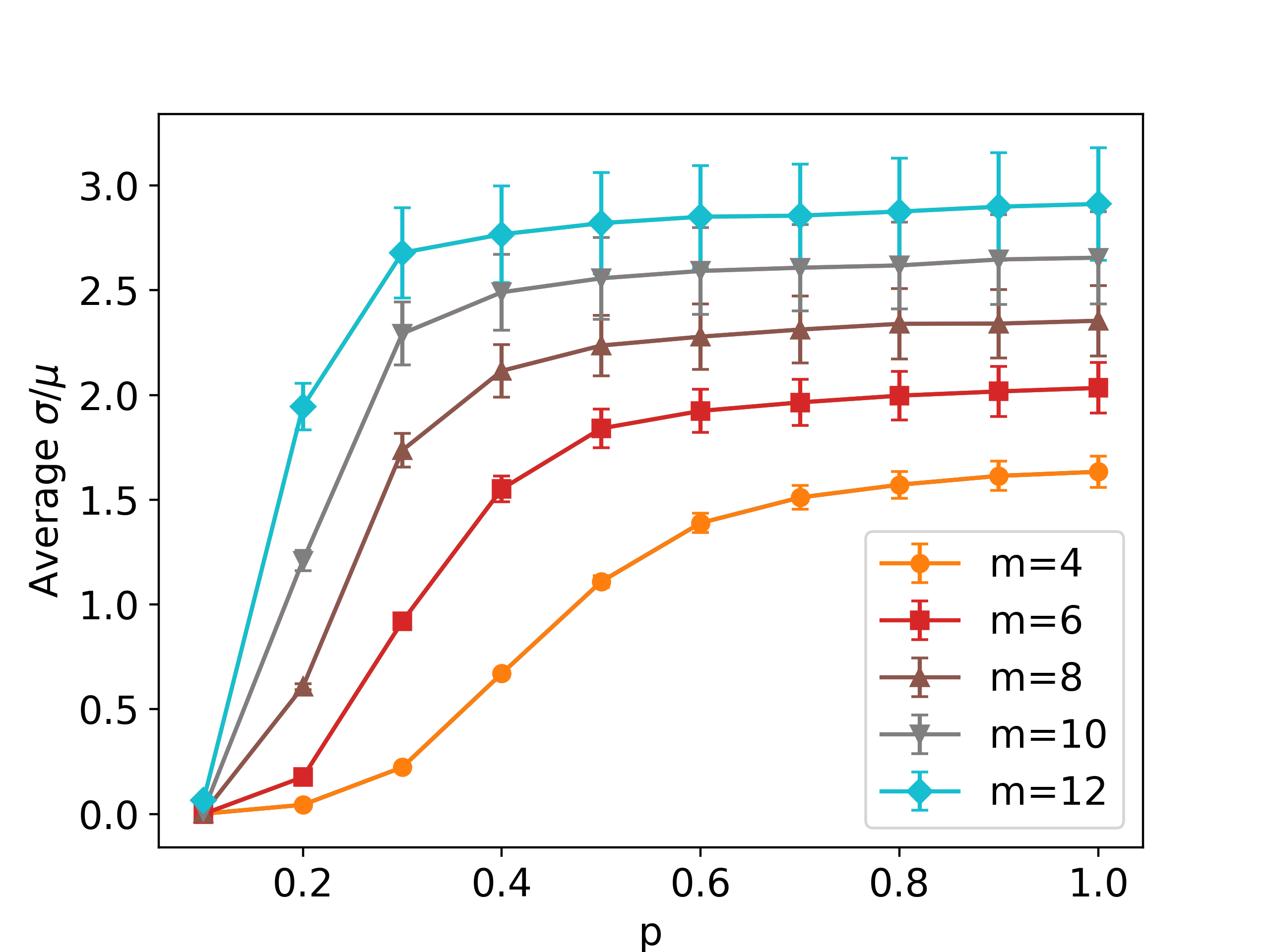}
    }
    \subfloat[]{
        \includegraphics[width=0.5\textwidth]{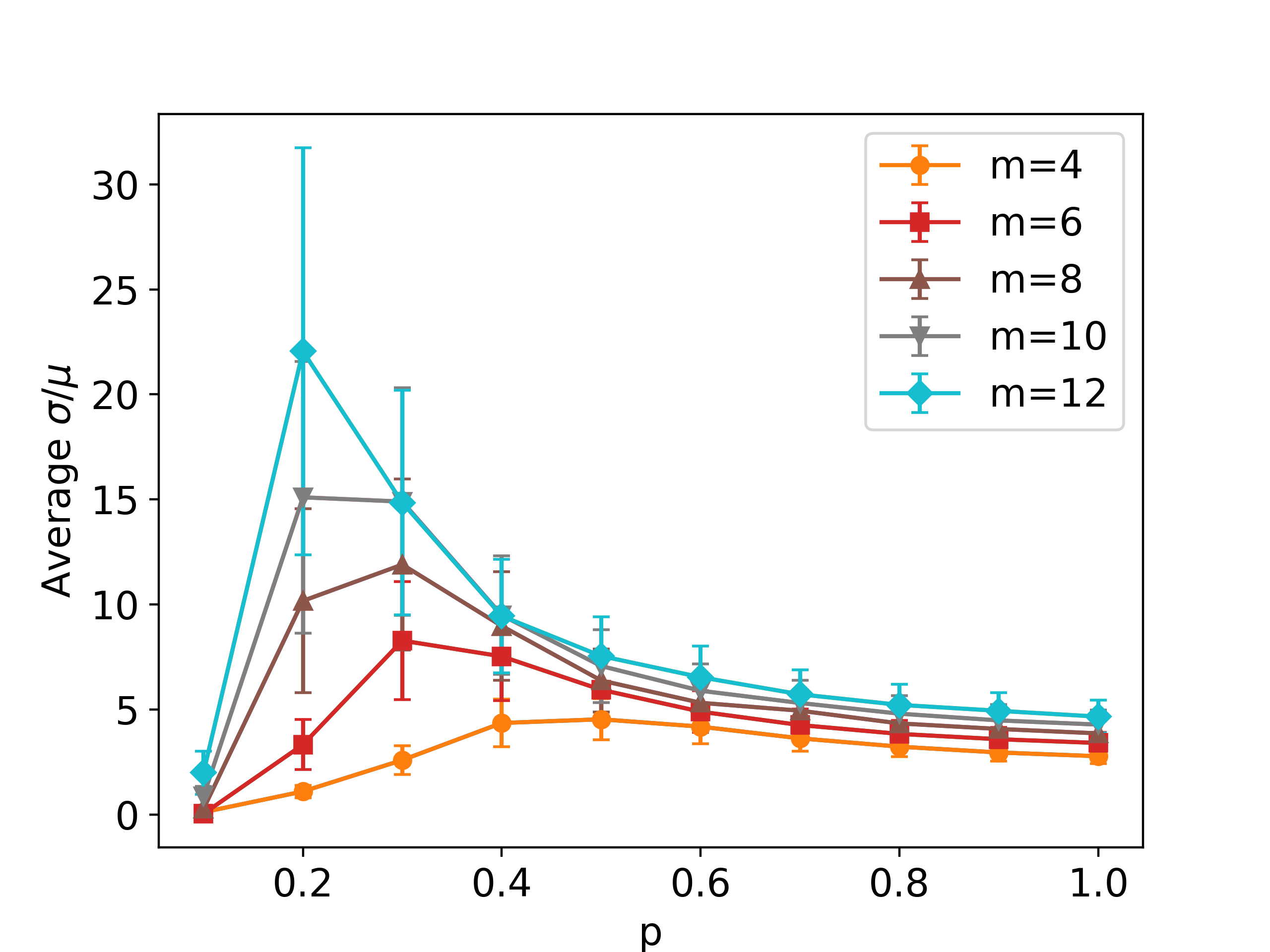}
    }
    \caption{Average relative standard deviation $\sigma / \mu$ of (a) Godsil-Gutman and (b) Barvinok estimators for random graphs. Error bars denote confidence interval of one sigma, evaluated by bootstrap resampling and averaged over the sampled graphs.}
    \label{fig:relsigma}
\end{figure*}

\begin{figure*}
    \centering
    \subfloat[]{
        \includegraphics[width=0.5\textwidth]{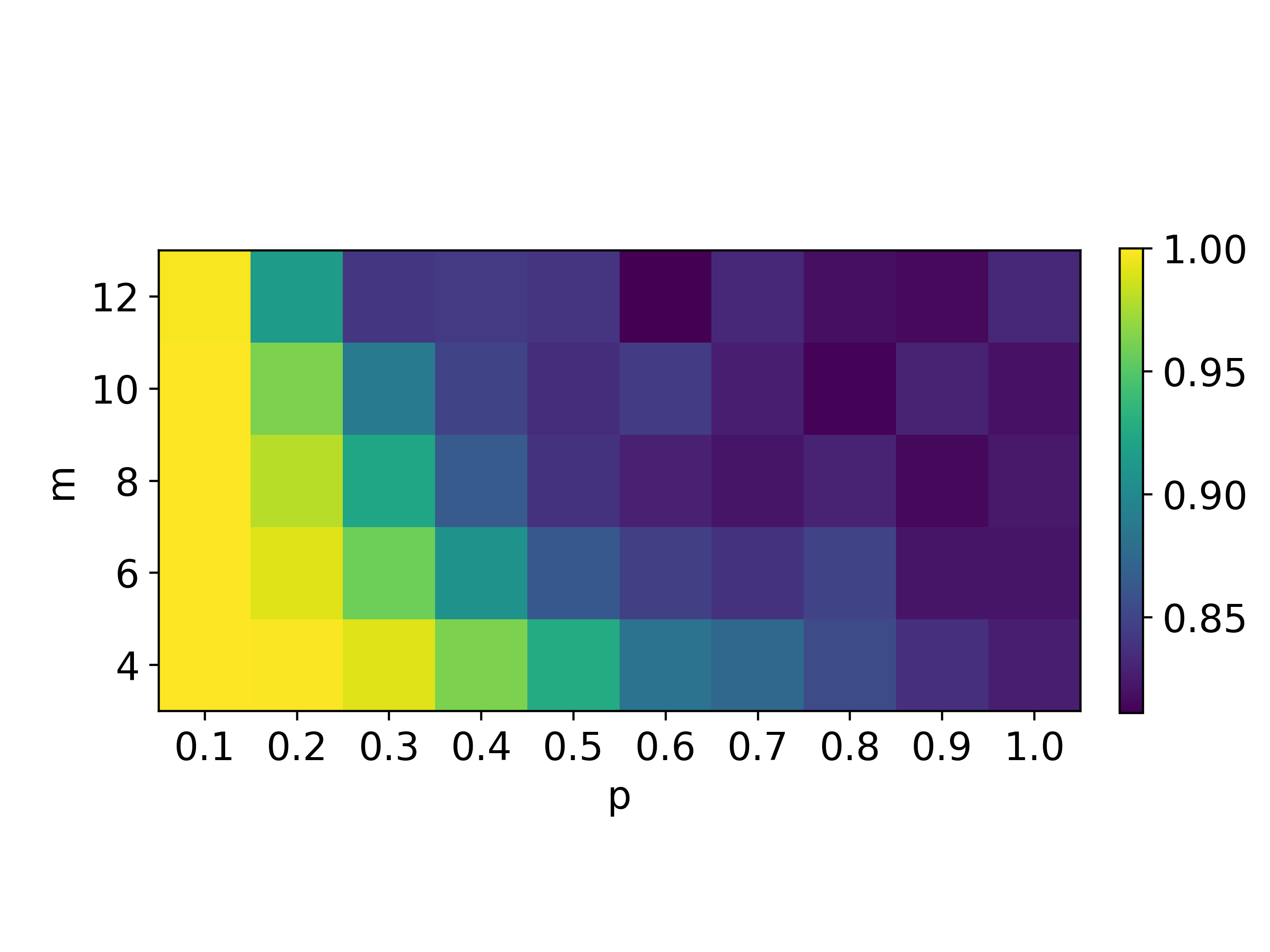}
    }
    \subfloat[]{
        \includegraphics[width=0.5\textwidth]{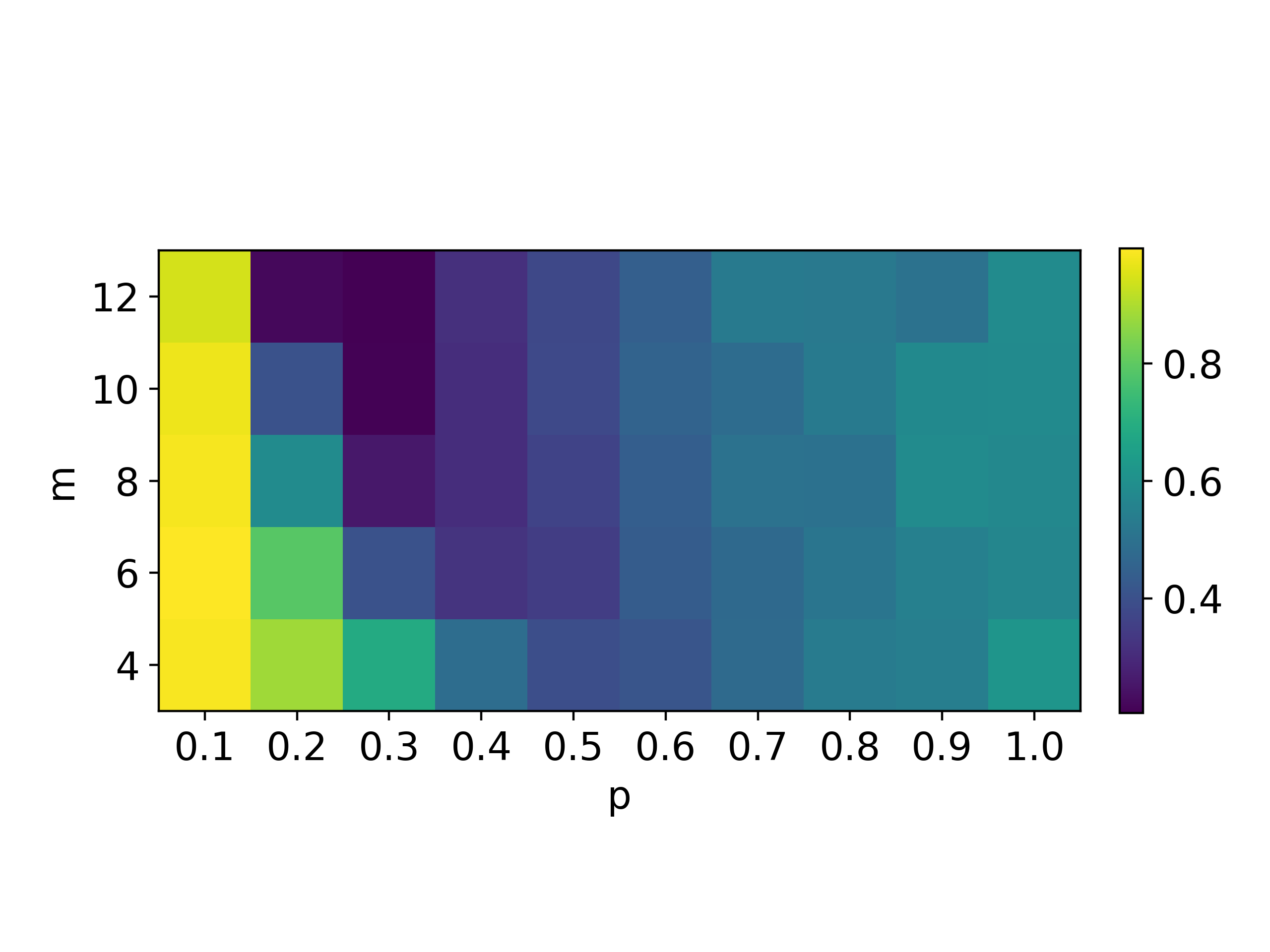}
    }
    \caption{Share of graph instances with Hafnian calculated up to relative error $0.05$ using (a)~Godsil-Gutman and~(b)~Barvinok estimators.}
    \label{fig:success_rate}
\end{figure*}

\begin{figure*}
    \centering
    \subfloat[]{
        \includegraphics[width=0.5\textwidth]{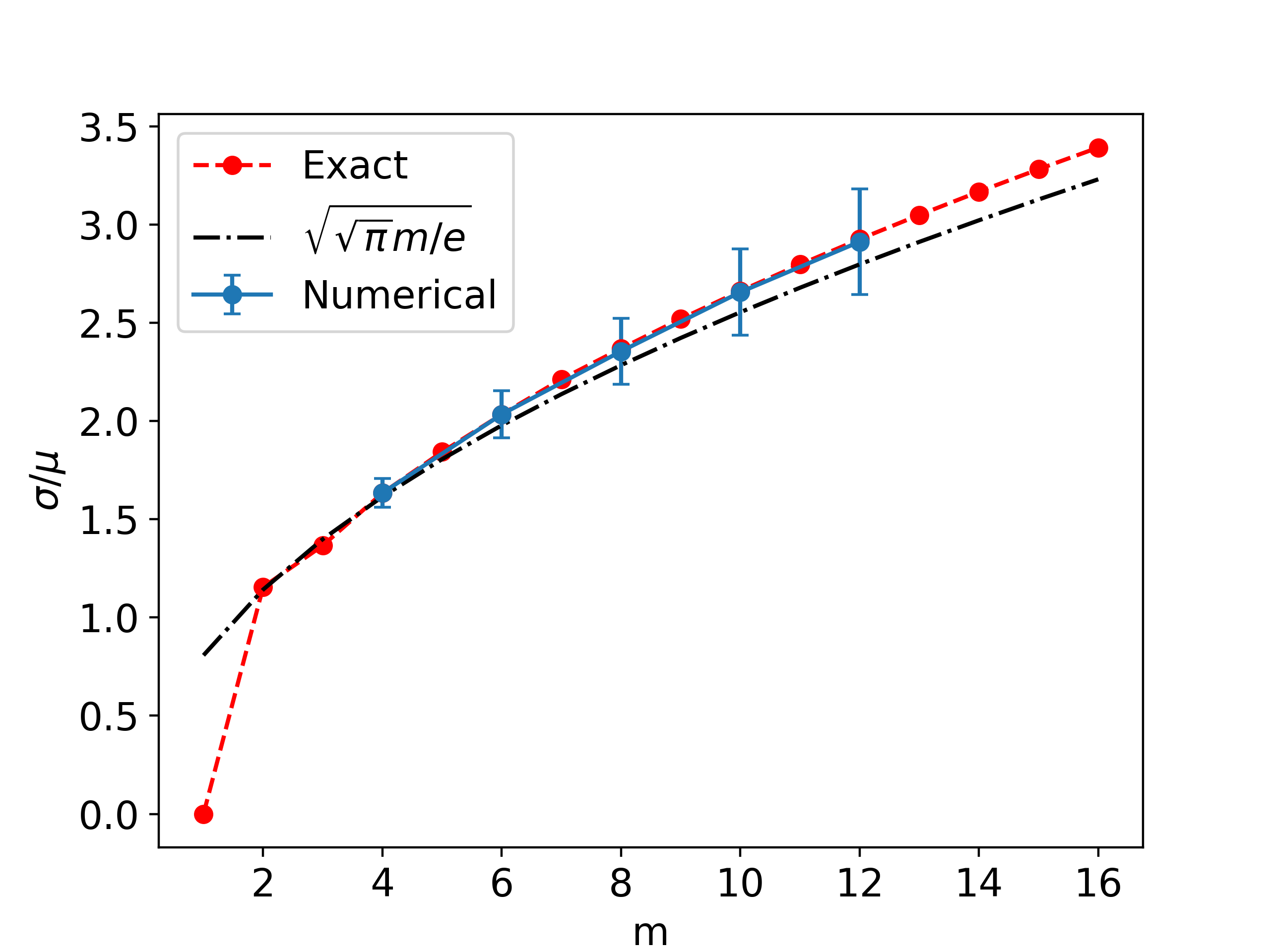}
    }
    \subfloat[]{
        \includegraphics[width=0.5\textwidth]{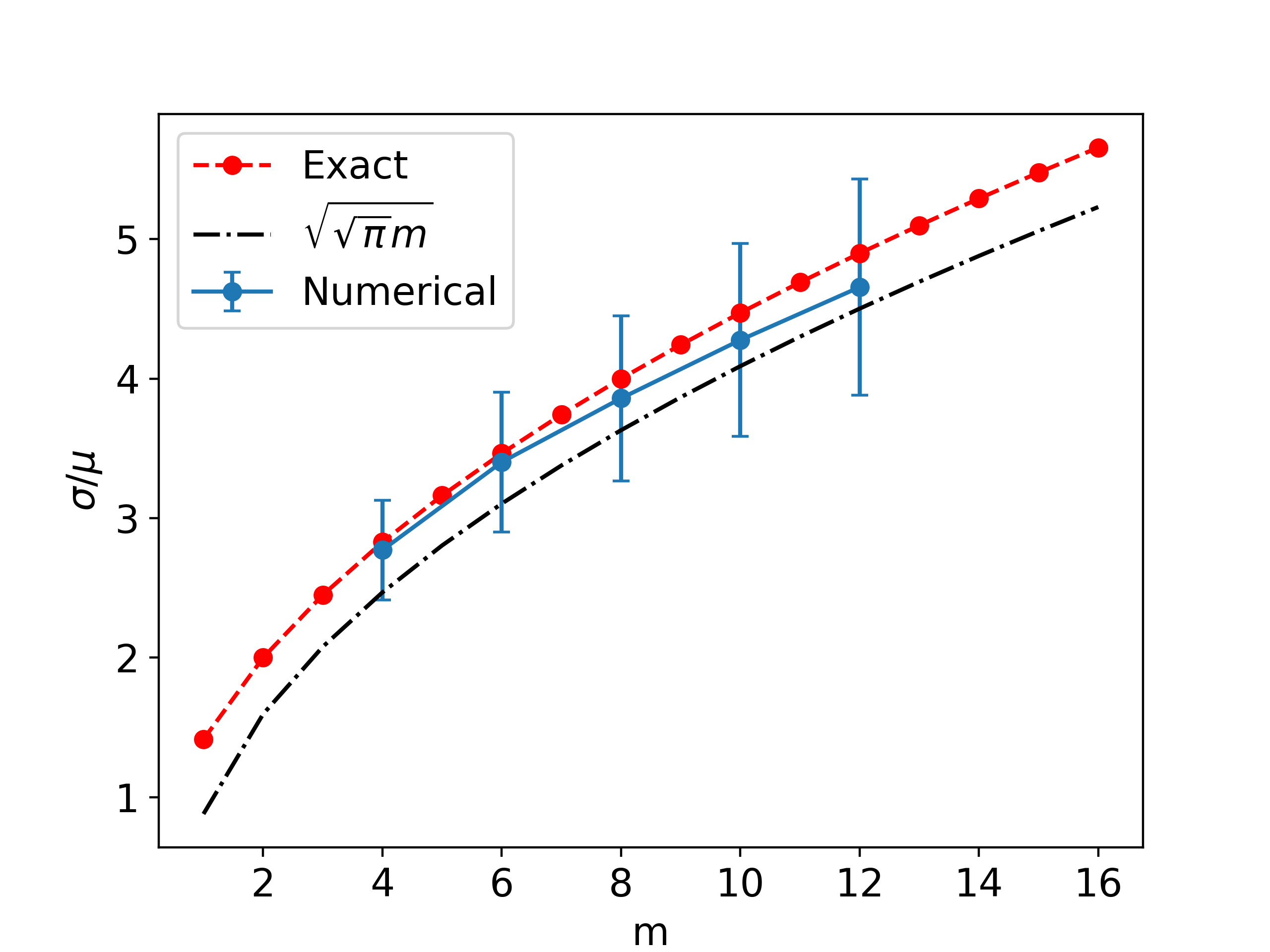}
    }
    \caption{Relative standard deviation $\sigma / \mu$ of (a) Godsil-Gutman and (b) Barvinok estimators for the complete graphs of different sizes. The dashed line shows the exact values computed using the Taylor expansion of (\ref{eq:egf_det2}), the dash-dotted line shows the asymptote as per Theorem~\ref{thm:complete_graph}.}
    \label{fig:relsigma_complete}
\end{figure*}


The sample variance alone might not contain the whole information about the estimate. For example, the disjoint graph discussed in Sec.~\ref{sec:estimator_variances} generates a distribution of determinants which has a high probability of observing zero and an extremely small probability of observing a very large value. In this case, a sample of all zeros would show zero relative sigma, which would be far from true. To account for that, we also studied the relative error of the estimator $|\bar{\mu} - \haf A| / \haf A$. Fig.~\ref{fig:success_rate} shows the share of graph instances for which the relative error is smaller than $0.05$. While there is some threshold behavior closely resembling that in Fig.~\ref{fig:relsigma}, overall there does not seem to be any situation where the Godsil-Gutman estimator would dramatically fail to estimate the Hafnian. Notably, for higher densities, the success rate appears to be constant with respect to $m$, confirming the analytical results.

\begin{figure*}
    \centering
    \includegraphics[width=\linewidth]{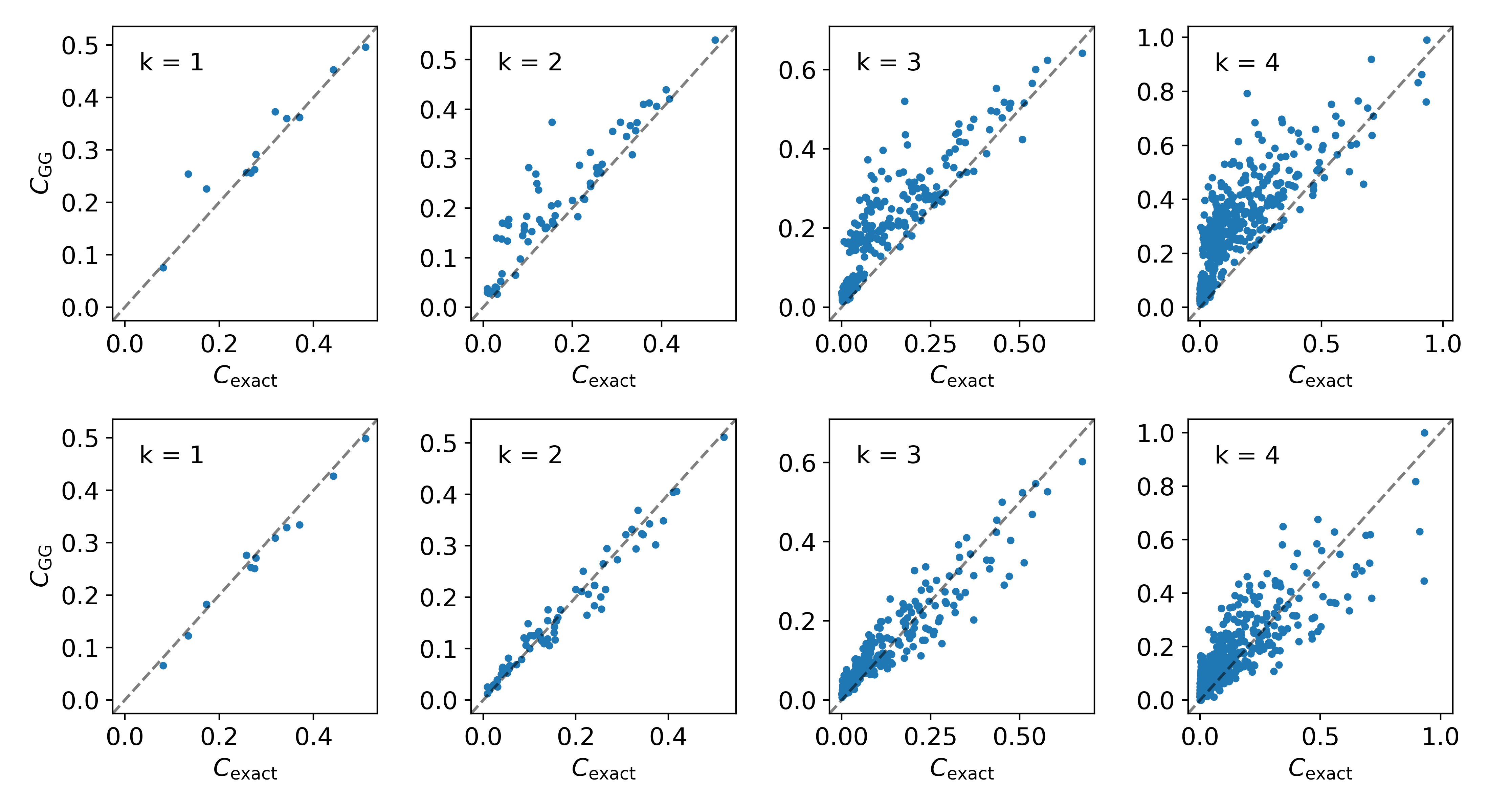}
    \caption{Correlation functions for the GBS experiment obtained with Godsil-Gutman sampling versus those obtained with exact sampling. Top row: $N=10$, bottom row: $N=1000$.}
    \label{fig:gbs_correlations}
\end{figure*}


Finally, we compared the sampling data for $p=1$, all of which correspond to the complete graphs of different sizes, to the analytical estimates of Theorem~\ref{thm:complete_graph}. As shown in Fig.~\ref{fig:relsigma_complete}, the estimated values of $\sigma / \mu$ show remarkable agreement with the theory. The asymptote $\sqrt{\sqrt{\pi} m e^\frac{\eta - 3}{2}}$ also closely follows the exact values.

\section{Classical simulation of nonnegative GBS}
\label{sec:numerics_gbs}

We now turn to testing the approximate Hafnian estimation within the GBS setting, 
which we briefly recall here. A lossless GBS experiment for an $m$-mode Gaussian state with zero means is fully specified by a symmetric $2m \times 2m$ matrix $\mathcal{A}$ with singular values $\sigma_i \in [0, 1)$, called the kernel matrix. Generally, the kernel matrix is complex, but here we restrict our attention to real and nonnegative matrices.

We consider the case of photon number resolving detectors, meaning that the detector can distinguish how many photons have arrived, as opposed to threshold detectors which can only tell whether any photons are present. The probability to observe an outcome $\mathbf{n} = (n_1, n_2, ..., n_m)$ is then equal to
\begin{equation}
    \label{eq:probability}
    P(\mathbf{n}) = \frac{1}{n_1! \dots n_m! \sqrt{\det Q}} \haf (\mathcal{A}_{\mathbf{n}}).
\end{equation}
Here $Q$ is the covariance matrix, expressed from $\mathcal{A}$ as follows:
\begin{equation}
    Q = \left(I - 
    \begin{pmatrix}
    0 & I \\ I & 0    
    \end{pmatrix} \mathcal{A} \right)^{-1}.
\end{equation}
The matrix $\mathcal{A}_{\mathbf{n}}$ is obtained from $\mathcal{A}$ in two steps. First, one constructs an intermediate matrix by taking the first column of $\mathcal{A}$ for $n_1$ times, the second column $n_2$ times, and so on until the $m$-th column, after which one repeats the column $i + m$ for $n_i$ times. Then, the matrix $\mathcal{A}_{\mathbf{n}}$ is constructed from the intermediate matrix by repeating the same procedure with the rows.

When the state is pure, the kernel matrix $\mathcal{A}$ is simplified to $\mathcal{A} = A \oplus A^*$. In this case, $A$ is also referred to as a kernel matrix. The kernel matrix of a pure state is connected to the squeezing values $r_i$ and the interferometer unitary $U$ by the Autonne-Takagi decomposition:
\begin{equation}
    A = U \operatorname{diag}(\tanh{r_1} \  \dots \  \tanh{r_m}) U^T.
\end{equation}
Graph-related applications of GBS typically work by setting $A$ to be equal to the adjacency matrix of the problem graph, up to appropriate scaling.

The method for the simulation of GBS proposed in Ref.~\cite{quesada_exact_2020} relies on the calculation of conditional probabilities of observing a particular number of photons in each waveguide. The method uses the calculation of Hafnian as a subroutine. When the kernel matrix is nonnegative, this subroutine can be replaced with the Barvinok or Godsil-Gutman estimator.

We simulated GBS for a kernel corresponding to the adjacency matrix of an Erd\"os-R\'enyi graph with $m=12$ and $p=0.5$. We produced 1000 samples, using the Godsil-Gutman estimator of the Hafnian ($N=10$ and $N=1000$ determinants per Hafnian estimation), as well as the exact Hafnian calculation. We capped the maximum number of photons per mode at $n_{max} = 4$. The kernel matrix was normalized so that the expected number of photons at $n_{max} = \infty$ would be equal to $\sqrt{m}$.

Comparing the sampled distributions directly in our case is difficult, since there are much more possible measurement outcomes than there are samples. For example, this means that we cannot correctly evaluate the Kullback-Leibler divergence or total variation distance. Instead, we study the low-order correlation functions $C(i_1, ..., i_k) = \langle n_{i_1} ... n_{i_k} \rangle$ for $k \in \{1, 2, 3, 4\}$. We denote the correlations $C_{\mathrm{GG}}$ when they are obtained from Godsil-Gutman sampling and $C_{\text{exact}}$ when they are obtained from exact sampling. The results are shown in Fig.~\ref{fig:gbs_correlations}. The samples obtained with $N=10$ demonstrate a noticeable deviation from their exact counterparts, while those obtained with $N=1000$ show good agreement with the exact sampling data. The deviation observed at $N=10$ appears to be biased in one direction: most correlations are overestimated. The populations of individual waveguides mostly turn out to be overestimated as well. It is not clear why the small sample size would produce this particular effect. However, this is likely connected to the distribution of the determinants in the Godsil-Gutman estimator. For most of the graphs we considered, the median Godsil-Gutman determinant is several times smaller than the mean Godsil-Gutman determinant. Consequently, a small sample is likely to underestimate the Hafnian, and overall this leads to a systematic shift in sampling probabilities.

\section{Conclusions}
\label{sec:conclude}

In this manuscript, we analyzed the behavior the Godsil-Gutman and Barvinok estimators of the Hafnian of a real nonnegative matrix. Despite the existence of particular examples where one needs an exponential number of samples to achieve a constant relative error, the overall performance of both methods for random graphs was surprisingly good. This shows that the usage of GBS for the densest $k$-subgraph problem is unlikely to provide an exponential speedup except possibly for some particular family of graphs.

The simulability of Boson Sampling in certain regimes has been extensively studied in the literature. Specifically, there is a substantial body of work dedicated to the simulation of Boson Sampling in imperfect conditions. In particular, it was found that Fock Boson Sampling with partially distinguishable photons~\cite{renema_efficient_2018-1}, as well as FBS with sufficient photon loss~\cite{renema_classical_2019}, can be simulated classically in polynomial time. 
There are approximate classical sampling algorithms for noisy GBS as well~\cite{villalonga_efficient_2022,popova_cracking_2022,oh_classical_2023-1}.
However, numerical experiments show that a noisy boson sampler is not much worse than a noiseless one when it is used to solve the $k$-densest subgraph problem~\cite{solomons_gaussian-boson-sampling-enhanced_2023}. Together with the efficient simulability of noisy GBS, this hints that the advantage derived from a quantum device is at best polynomial. In addition, Ref.~\cite{oh_quantum-inspired_2023-1} proposes a quantum-inspired classical algorithm for the $k$-densest subgraph problem and shows that the performance is also comparable to quantum sampling. 

In conclusion, our work, together with the evidence from the existing literature, suggests (though does not yet decisively prove) that the ability to perform GBS with a nonnegative kernel does not provide quantum advantage. We suspect that there is a method to deal with the special cases where the Godsil-Gutman estimator fails, and that by finding it one can construct a fully polynomial randomized approximation scheme (FPRAS) for the Hafnian of a nonnegative matrix.


\acknowledgments

The work of the authors was supported by Rosatom in the framework of the Roadmap for Quantum Computing (Contract No.~868-1.3-15/15-2021 dated October~5,~2021 and Contract No.~R2320 dated March~09,~2023).
The exact computation of Hafnians was performed using \texttt{The~Walrus}~\cite{gupt_walrus_2019}.

\appendix

\section{Example for $m=2$}
\label{sec:det2_example}

Let us consider the case when $A, W$, and $G$ are $4 \times 4$ matrices:

\begin{gather}
    A = \begin{pmatrix}
    0 & a_{12} & a_{13} & a_{14} \\
    a_{12} & 0 & a_{23} & a_{24} \\
    a_{13} & a_{23} & 0 & a_{34} \\
    a_{14} & a_{24} & a_{34} & 0 \\
    \end{pmatrix}, \\
    W = \begin{pmatrix}
    0 & w_{12} & w_{13} & w_{14} \\
    -w_{12} & 0 & w_{23} & w_{24} \\
    -w_{13} & -w_{23} & 0 & w_{34} \\
    -w_{14} & -w_{24} & -w_{34} & 0 \\
    \end{pmatrix}, \\
    G = 
    \begin{pmatrix}
    0 & g_{12} & g_{13} & g_{14} \\
    g_{12} & 0 & g_{23} & g_{24} \\
    g_{13} & g_{23} & 0 & g_{34} \\
    g_{14} & g_{24} & g_{34} & 0 \\
    \end{pmatrix} = \\    
    =\begin{pmatrix}
    0 & w_{12}\sqrt{a_{12}} & w_{13}\sqrt{a_{13}} & w_{14}\sqrt{a_{14}} \\
    -w_{12}\sqrt{a_{12}} & 0 & w_{23}\sqrt{a_{23}} & w_{24}\sqrt{a_{24}} \\
    -w_{13}\sqrt{a_{13}} & -w_{23}\sqrt{a_{23}} & 0 & w_{34}\sqrt{a_{34}} \\
    -w_{14}\sqrt{a_{14}} & -w_{24}\sqrt{a_{24}} & -w_{34}\sqrt{a_{34}} & 0 \\
    \end{pmatrix}.
\end{gather}
The determinant of $G$ is equal to
\begin{multline}
    \det G = g_{12}^2 g_{34}^2 + g_{13}^2 g_{24}^2 + g_{14}^2 g_{23}^2  - \\ - 2 g_{12} g_{34} g_{13} g_{24} - 2 g_{12} g_{14} g_{23} g_{34} - 2 g_{13} g_{24} g_{14} g_{23}.
\end{multline}
Taking the square of this and grouping the terms where all $g_{ij}$ are in the even power, we get:
\begin{multline}
    (\det G)^2 = g_{12}^4 g_{34}^4 + g_{13}^4 g_{24}^4
    + g_{14}^4 g_{23}^4 +\\
    + 6 g_{12}^2 g_{34}^2 g_{13}^2 g_{24}^2 
    + 6 g_{12}^2 g_{14}^2 g_{23}^2 g_{34}^2 + \\
    + 6 g_{13}^2 g_{24}^2 g_{14}^2 g_{23}^2 + \text{(odd-power terms)}.
\end{multline}
Taking the expectation value of this expression, we obtain
\begin{multline}
    \mathbb{E}(\det G)^2 = \eta^2 a_{12}^2 a_{34}^2 + \eta^2 a_{13}^2 a_{24}^2 + \eta^2 a_{14}^2 a_{23}^2 + \\
    + 6 a_{12} a_{34} a_{13} a_{24} 
    + 6 a_{12} a_{14} a_{23} a_{34} + \\
    + 6 a_{13} a_{24} a_{14} a_{23}. 
\end{multline}

Diagrammatically we can show this as follows:
\begin{multline}
    \mathbb{E}(\det G)^2 = \eta^2 \parbox{0.1\linewidth}{\includegraphics[width=\linewidth]{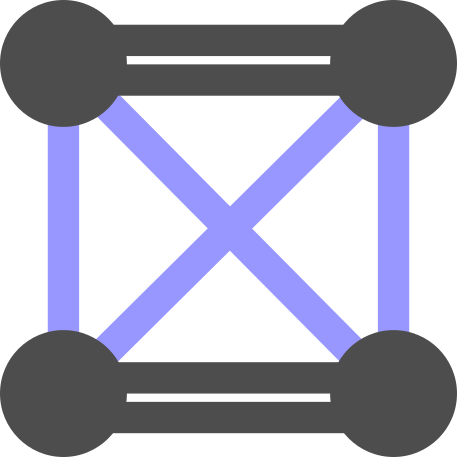}} 
    + \eta^2 \parbox{0.1\linewidth}{\includegraphics[width=\linewidth]{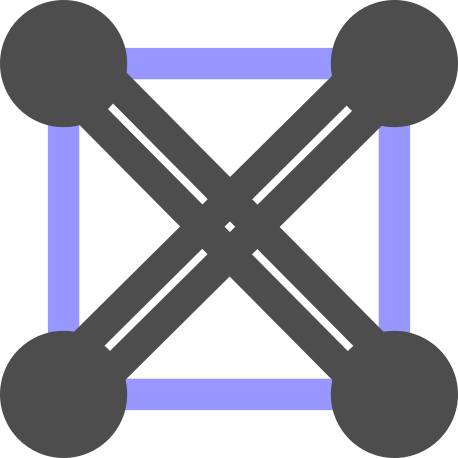}} 
    + \eta^2 \parbox{0.1\linewidth}{\includegraphics[width=\linewidth]{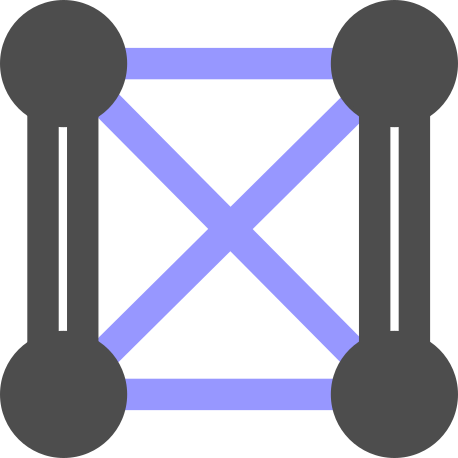}} + \\
    + 6 \parbox{0.1\linewidth}{\includegraphics[width=\linewidth]{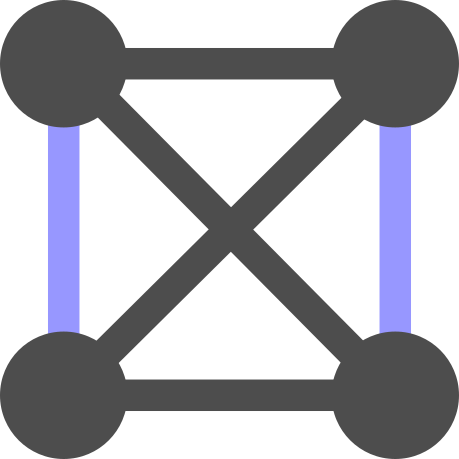}} 
    + 6 \parbox{0.1\linewidth}{\includegraphics[width=\linewidth]{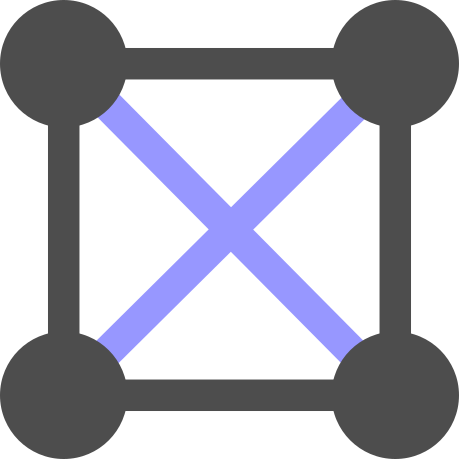}}
    + 6 \parbox{0.1\linewidth}{\includegraphics[width=\linewidth]{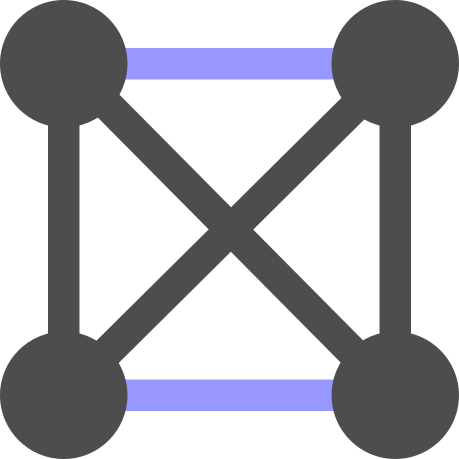}}. 
\end{multline}
Here the vertices are enumerated in clockwise order starting from top left. The edges in the first three diagrams are doubled meaning that we take the squares of the respective matrix elements.

We can observe that the squared Hafnian has a similar structure. The Hafnian of $A$ is equal to
\begin{equation}
    \haf A = a_{12} a_{34} + a_{13} a_{24} + a_{14} a_{23}.
\end{equation}
Taking the square of that and grouping the terms together, we get
\begin{multline}
    (\haf A)^2 = a_{12}^2 a_{34}^2 + a_{13}^2 a_{24}^2 + a_{14}^2 a_{23}^2 + \\
    + 2 a_{12} a_{34} a_{13} a_{24} 
    + 2 a_{12} a_{14} a_{23} a_{34} + \\
    + 2 a_{13} a_{24} a_{14} a_{23}. 
\end{multline}
Or, diagrammatically:
\begin{multline}
    (\haf A)^2 = \parbox{0.1\linewidth}{\includegraphics[width=\linewidth]{12_34.png}} 
    + \parbox{0.1\linewidth}{\includegraphics[width=\linewidth]{13_24.png}} 
    + \parbox{0.1\linewidth}{\includegraphics[width=\linewidth]{14_23.png}} + \\
    + 2 \parbox{0.1\linewidth}{\includegraphics[width=\linewidth]{hourglass.png}} 
    + 2 \parbox{0.1\linewidth}{\includegraphics[width=\linewidth]{cycle.png}}
    + 2 \parbox{0.1\linewidth}{\includegraphics[width=\linewidth]{side_hourglass.png}}. 
\end{multline}

\section{Proof of Theorem \ref{thm:complete_graph}}
\label{sec:species}


Before we prove the main result, we need to introduce a few definitions.

A species $F$ is a functor from the category of finite sets and bijections to itself~\cite{joyal_theorie_1981,bergeron_combinatorial_1997}. If $S$ is a finite set, then the elements of $F(S)$ are referred to as $F$-structures. For example, we can define a functor $\mathbf{Match}$ which maps a finite set to the set of all perfect matchings on that set. In this case, the perfect matchings are the $F$-structures. 

With each species $F$ we can associate an exponential generating function:
\begin{equation}
    F(x) = \sum_{n = 0}^{\infty} \frac{1}{n!} |F(\{1, ..., n\})| x^n.
\end{equation}
The requirement that $F$ is a functor guarantees that the cardinality of $F(X)$ only depends on the cardinality of $X$.

Crucially, there is a correspondence between the operations on the species and operations on the generating functions. We will use two of these operations: product and substitution. Let $F, G$ be two species. The product species $F \cdot G$ is a functor that maps a set $U$ to the set
\begin{equation}
    \bigsqcup_{S \subseteq U} F(S) \times G(U \setminus S).
\end{equation}
The generating function of $F \cdot G$ is the product $F(x)G(x)$ of the corresponding generating functions. The substitution $F \circ G$ is a more complicated functor. The structures created by $F \circ G$ are partitions of the finite set, each endowed with a $G$-structure, forming together an $F$-structure. It acts on sets as follows:
\begin{equation}
    (F \circ G) (S) = \bigsqcup_{\pi \in \mathrm{part}(S)} \left( F(\pi) \times \prod_{Y \in \pi} G(Y) \right).
\end{equation}
The generating function of $F \circ G$ is $F(G(x))$.

Finally, we need to introduce the weighted sets and weighted species. Let $K$ be a ring. A $K$-weighted set is a pair $(S, w)$, where $S$ is a finite set, and $w: S \rightarrow K$ is a function called weight. Given two weighted sets $(A, u), (B, v)$, the product $A \times B$ can be equipped with a weight by setting $w((a, b)) = u(a) \cdot v(b)$. The disjoint union $A \sqcup B$ can be equipped with a weight by taking the weights of the comprising subsets.

A weighted species is a functor from the category of finite sets and bijections to the category of weighted sets and weight-preserving bijections. Given a weighted species $F$, we can associate with it an exponential generating function with coefficients in $K$:
\begin{equation}
    F(x) = \sum_{n = 0}^{\infty} \frac{1}{n!} |F(\{1, ..., n\})|_w x^n.
\end{equation}
Here $|F(S)|_w = \sum_{x \in F(S)} w(x)$ is the weighted sum of the elements in $F(S)$. The operations on weighted species and their generating functions correspond to each other in exactly the same way as non-weighted species and their generating functions.

With this toolset at our disposal, we can construct a weighted species whose structures on $2m$-element sets are vertex covers consisting of matchings and cycles of even length, with weights equal to the coefficients in (\ref{eq:expected_det2}). All of these structures are present in a complete graph, so what we need is the sum of their weights. Thus, the coefficient before $x^{2m}$ in the generating function of this species will represent the second moment of the probabilistic estimator.

We consider the ring of real polynomials in one variable $\mathbb{R}[t]$ as the ring of weights. We will later set $t = \eta$ to account for the type of the estimator. Let us define the following functors (when weight is not specified, it is equal to one on all elements of the set):

\begin{enumerate}
    \item $\mathbf{Pair}_t$ maps any two-element set to $(\{\bullet \}, t)$, where $\{\bullet \}$ is a singleton set. All other sets are mapped to the empty set. The generating function is equal to $\mathbf{Pair}_t(x) = \frac{tx^2}{2}$. 
    \item $\mathbf{Exp}$ maps any set to the singleton set. Its generating function is equal to $e^x$.
    \item $\mathbf{Cycle}$ maps any set of size $4 + 2k$ to the set of all cyclic orderings on this set (up to the inverse), and all other sets to the empty set. Its generating function is equal to
    \begin{multline}
        \mathbf{Cycle}(x) = \frac12 \left( \frac{3!x^4}{4!} + \frac{5!x^6}{6!} + ... \right) =\\
        = -\frac14 \log(1 - x^2) - \frac{x^2}{4}.
    \end{multline}
    We will need a weighted functor $6\mathbf{Cycle}$, such that the weights of all elements of the output sets are equal to~6. Its generating function is thus equal to $-\frac32 \log(1 - x^2) - \frac{3x^2}{2}$.
\end{enumerate}

We will construct two more functors: $\mathbf{Match}_t$, which maps a set to its perfect matchings with weight $t$ raised to the power of the size of the matching, and $\mathbf{CycleCover}$, which maps a set to the set of its covers with cycles of even length.

Since every subgraph involved in (\ref{eq:expected_det2}) consists of a disjoint union of matchings and cycle covers, the subgraphs for the complete graph are generated by the functor $\mathbf{Match}_t \cdot \mathbf{CycleCover}$. Both the perfect matchings and cycle covers of a set can be thought of as partitions of the set together with the corresponding structure on each part. In terms of species, this means that $\mathbf{Match}_t = \mathbf{Exp} \circ \mathbf{Pair}_t$ and $\mathbf{CycleCover} = \mathbf{Exp} \circ 6\mathbf{Cycle}$. Altogether this means that the generating function for the second moments (\ref{eq:expected_det2}) is the following:
\begin{equation}
\label{eq:egf_det2}
    f(x, t) = e^{\frac{tx^2}{2}}e^{-\frac{3}{2} \log(1 - x^2) - \frac{3x^2}{2}} = \frac{\exp(\frac{(t-3)x^2}{2})}{(1 - x^2)^{3/2}}.
\end{equation}

Knowing the generating function, we can conduct the analysis of asymptotic behavior of its series expansion using the Darboux's method~\cite{wilf2005generatingfunctionology,flajolet2009analytic}. First, define $\alpha = (t - 3) / 2$ and $y = x^2$. Define $g(y) = f(x)$:
\begin{equation}
    g(y) = \frac{e^{\alpha y}}{(1 - y)^{3/2}}.
\end{equation}
Expanding the numerator around $y = 1$, we obtain the following:
\begin{multline}
    g(y) = e^\alpha \left( (1 - y)^{-3/2} - \alpha (1 - y)^{-1/2} \right. + \\
    \left. + \frac{\alpha^2}{2} \sqrt{1-y} + R(y) \right).
\end{multline}
Here $R(y)$ contains all the remaining terms of the Taylor expansion and is once continuously differentiable at $|y|= 1$. The $n$-th term in the series expansion of $g(y)$ is then equal to
\begin{multline}
    [y^n] g(y) = e^\alpha \left( \binom{n + 1/2}{n} - \alpha \binom{n - 1/2}{n} + \right. \\
    \left. + \frac{\alpha^2}{2} \binom{n - 3/2}{n}+ o\left(\frac{1}{n}\right) \right).
\end{multline}
Here the last term of the expansion is vanishing faster than $1 / n$ due to the Darboux's theorem. Expanding the binomial coefficients and using the Striling's formula for the gamma function, we obtain:
\begin{equation}
    [y^n] g(y) = \frac{2}{\sqrt{\pi}}e^\alpha \sqrt{n} + O\left(\sqrt{\frac{1}{n}}\right).
\end{equation}
This means that the expected square of the determinant (\ref{eq:expected_det2}) for a complete graph on $2m$ vertices tends to $\frac{2}{\sqrt{\pi}}e^\alpha \sqrt{m} (2m)!$ as $m$ goes to infinity. To calculate the relative variance, we recall that the number of perfect matchings in a complete graph is equal to $(2m - 1)!!$. Gathering all the terms and using once again the Stirling's formula for the factorials, we obtain:
\begin{multline}
    \frac{\mathbb{E} \det G^2}{(\haf A)^2} = \frac{(2m)!(\frac{2}{\sqrt{\pi}}e^\alpha \sqrt{m} + O(1/\sqrt{m}))}{((2m-1)!!)^2} = \\
    =  \sqrt{\pi} m e^{\alpha} + O(1).
    \label{eq:asymptote}
\end{multline}
Thus, the relative error of either estimator grows as $\frac{\sqrt{m}}{\sqrt{N}}$, which implies that the number of samples $N$ should grow linearly with the size of the graph to preserve the same relative tolerance. Surprisingly enough, the fourth moment $\eta$ only changes the constant prefactor and not the overall scaling.

\bibliography{refs}

\begin{thebibliography}{29}%
\makeatletter
\providecommand \@ifxundefined [1]{%
 \@ifx{#1\undefined}
}%
\providecommand \@ifnum [1]{%
 \ifnum #1\expandafter \@firstoftwo
 \else \expandafter \@secondoftwo
 \fi
}%
\providecommand \@ifx [1]{%
 \ifx #1\expandafter \@firstoftwo
 \else \expandafter \@secondoftwo
 \fi
}%
\providecommand \natexlab [1]{#1}%
\providecommand \enquote  [1]{``#1''}%
\providecommand \bibnamefont  [1]{#1}%
\providecommand \bibfnamefont [1]{#1}%
\providecommand \citenamefont [1]{#1}%
\providecommand \href@noop [0]{\@secondoftwo}%
\providecommand \href [0]{\begingroup \@sanitize@url \@href}%
\providecommand \@href[1]{\@@startlink{#1}\@@href}%
\providecommand \@@href[1]{\endgroup#1\@@endlink}%
\providecommand \@sanitize@url [0]{\catcode `\\12\catcode `\$12\catcode
  `\&12\catcode `\#12\catcode `\^12\catcode `\_12\catcode `\%12\relax}%
\providecommand \@@startlink[1]{}%
\providecommand \@@endlink[0]{}%
\providecommand \url  [0]{\begingroup\@sanitize@url \@url }%
\providecommand \@url [1]{\endgroup\@href {#1}{\urlprefix }}%
\providecommand \urlprefix  [0]{URL }%
\providecommand \Eprint [0]{\href }%
\providecommand \doibase [0]{https://doi.org/}%
\providecommand \selectlanguage [0]{\@gobble}%
\providecommand \bibinfo  [0]{\@secondoftwo}%
\providecommand \bibfield  [0]{\@secondoftwo}%
\providecommand \translation [1]{[#1]}%
\providecommand \BibitemOpen [0]{}%
\providecommand \bibitemStop [0]{}%
\providecommand \bibitemNoStop [0]{.\EOS\space}%
\providecommand \EOS [0]{\spacefactor3000\relax}%
\providecommand \BibitemShut  [1]{\csname bibitem#1\endcsname}%
\let\auto@bib@innerbib\@empty
\bibitem [{\citenamefont {Hamilton}\ \emph {et~al.}(2017)\citenamefont
  {Hamilton}, \citenamefont {Kruse}, \citenamefont {Sansoni}, \citenamefont
  {Barkhofen}, \citenamefont {Silberhorn},\ and\ \citenamefont
  {Jex}}]{hamilton_gaussian_2017}%
  \BibitemOpen
  \bibfield  {author} {\bibinfo {author} {\bibfnamefont {C.~S.}\ \bibnamefont
  {Hamilton}}, \bibinfo {author} {\bibfnamefont {R.}~\bibnamefont {Kruse}},
  \bibinfo {author} {\bibfnamefont {L.}~\bibnamefont {Sansoni}}, \bibinfo
  {author} {\bibfnamefont {S.}~\bibnamefont {Barkhofen}}, \bibinfo {author}
  {\bibfnamefont {C.}~\bibnamefont {Silberhorn}},\ and\ \bibinfo {author}
  {\bibfnamefont {I.}~\bibnamefont {Jex}},\ }\bibfield  {title} {\bibinfo
  {title} {Gaussian {{Boson Sampling}}},\ }\href
  {https://doi.org/10.1103/PhysRevLett.119.170501} {\bibfield  {journal}
  {\bibinfo  {journal} {Physical Review Letters}\ }\textbf {\bibinfo {volume}
  {119}},\ \bibinfo {pages} {170501} (\bibinfo {year} {2017})}\BibitemShut
  {NoStop}%
\bibitem [{\citenamefont {Kruse}\ \emph {et~al.}(2019)\citenamefont {Kruse},
  \citenamefont {Hamilton}, \citenamefont {Sansoni}, \citenamefont {Barkhofen},
  \citenamefont {Silberhorn},\ and\ \citenamefont {Jex}}]{kruse_detailed_2019}%
  \BibitemOpen
  \bibfield  {author} {\bibinfo {author} {\bibfnamefont {R.}~\bibnamefont
  {Kruse}}, \bibinfo {author} {\bibfnamefont {C.~S.}\ \bibnamefont {Hamilton}},
  \bibinfo {author} {\bibfnamefont {L.}~\bibnamefont {Sansoni}}, \bibinfo
  {author} {\bibfnamefont {S.}~\bibnamefont {Barkhofen}}, \bibinfo {author}
  {\bibfnamefont {C.}~\bibnamefont {Silberhorn}},\ and\ \bibinfo {author}
  {\bibfnamefont {I.}~\bibnamefont {Jex}},\ }\bibfield  {title} {\bibinfo
  {title} {A detailed study of {{Gaussian Boson Sampling}}},\ }\href
  {https://doi.org/10.1103/PhysRevA.100.032326} {\bibfield  {journal} {\bibinfo
   {journal} {Physical Review A}\ }\textbf {\bibinfo {volume} {100}},\ \bibinfo
  {pages} {032326} (\bibinfo {year} {2019})}\BibitemShut {NoStop}%
\bibitem [{\citenamefont {Zhong}\ \emph {et~al.}(2020)\citenamefont {Zhong},
  \citenamefont {Wang}, \citenamefont {Deng}, \citenamefont {Chen},
  \citenamefont {Peng}, \citenamefont {Luo}, \citenamefont {Qin}, \citenamefont
  {Wu}, \citenamefont {Ding}, \citenamefont {Hu}, \citenamefont {Hu},
  \citenamefont {Yang}, \citenamefont {Zhang}, \citenamefont {Li},
  \citenamefont {Li}, \citenamefont {Jiang}, \citenamefont {Gan}, \citenamefont
  {Yang}, \citenamefont {You}, \citenamefont {Wang}, \citenamefont {Li},
  \citenamefont {Liu}, \citenamefont {Lu},\ and\ \citenamefont
  {Pan}}]{zhong_quantum_2020}%
  \BibitemOpen
  \bibfield  {author} {\bibinfo {author} {\bibfnamefont {H.-S.}\ \bibnamefont
  {Zhong}}, \bibinfo {author} {\bibfnamefont {H.}~\bibnamefont {Wang}},
  \bibinfo {author} {\bibfnamefont {Y.-H.}\ \bibnamefont {Deng}}, \bibinfo
  {author} {\bibfnamefont {M.-C.}\ \bibnamefont {Chen}}, \bibinfo {author}
  {\bibfnamefont {L.-C.}\ \bibnamefont {Peng}}, \bibinfo {author}
  {\bibfnamefont {Y.-H.}\ \bibnamefont {Luo}}, \bibinfo {author} {\bibfnamefont
  {J.}~\bibnamefont {Qin}}, \bibinfo {author} {\bibfnamefont {D.}~\bibnamefont
  {Wu}}, \bibinfo {author} {\bibfnamefont {X.}~\bibnamefont {Ding}}, \bibinfo
  {author} {\bibfnamefont {Y.}~\bibnamefont {Hu}}, \bibinfo {author}
  {\bibfnamefont {P.}~\bibnamefont {Hu}}, \bibinfo {author} {\bibfnamefont
  {X.-Y.}\ \bibnamefont {Yang}}, \bibinfo {author} {\bibfnamefont {W.-J.}\
  \bibnamefont {Zhang}}, \bibinfo {author} {\bibfnamefont {H.}~\bibnamefont
  {Li}}, \bibinfo {author} {\bibfnamefont {Y.}~\bibnamefont {Li}}, \bibinfo
  {author} {\bibfnamefont {X.}~\bibnamefont {Jiang}}, \bibinfo {author}
  {\bibfnamefont {L.}~\bibnamefont {Gan}}, \bibinfo {author} {\bibfnamefont
  {G.}~\bibnamefont {Yang}}, \bibinfo {author} {\bibfnamefont {L.}~\bibnamefont
  {You}}, \bibinfo {author} {\bibfnamefont {Z.}~\bibnamefont {Wang}}, \bibinfo
  {author} {\bibfnamefont {L.}~\bibnamefont {Li}}, \bibinfo {author}
  {\bibfnamefont {N.-L.}\ \bibnamefont {Liu}}, \bibinfo {author} {\bibfnamefont
  {C.-Y.}\ \bibnamefont {Lu}},\ and\ \bibinfo {author} {\bibfnamefont {J.-W.}\
  \bibnamefont {Pan}},\ }\bibfield  {title} {\bibinfo {title} {Quantum
  computational advantage using photons},\ }\href
  {https://doi.org/10.1126/science.abe8770} {\bibfield  {journal} {\bibinfo
  {journal} {Science}\ }\textbf {\bibinfo {volume} {370}},\ \bibinfo {pages}
  {1460} (\bibinfo {year} {2020})}\BibitemShut {NoStop}%
\bibitem [{\citenamefont {Zhong}\ \emph {et~al.}(2021)\citenamefont {Zhong},
  \citenamefont {Deng}, \citenamefont {Qin}, \citenamefont {Wang},
  \citenamefont {Chen}, \citenamefont {Peng}, \citenamefont {Luo},
  \citenamefont {Wu}, \citenamefont {Gong}, \citenamefont {Su}, \citenamefont
  {Hu}, \citenamefont {Hu}, \citenamefont {Yang}, \citenamefont {Zhang},
  \citenamefont {Li}, \citenamefont {Li}, \citenamefont {Jiang}, \citenamefont
  {Gan}, \citenamefont {Yang}, \citenamefont {You}, \citenamefont {Wang},
  \citenamefont {Li}, \citenamefont {Liu}, \citenamefont {Renema},
  \citenamefont {Lu},\ and\ \citenamefont
  {Pan}}]{zhong_phase-programmable_2021}%
  \BibitemOpen
  \bibfield  {author} {\bibinfo {author} {\bibfnamefont {H.-S.}\ \bibnamefont
  {Zhong}}, \bibinfo {author} {\bibfnamefont {Y.-H.}\ \bibnamefont {Deng}},
  \bibinfo {author} {\bibfnamefont {J.}~\bibnamefont {Qin}}, \bibinfo {author}
  {\bibfnamefont {H.}~\bibnamefont {Wang}}, \bibinfo {author} {\bibfnamefont
  {M.-C.}\ \bibnamefont {Chen}}, \bibinfo {author} {\bibfnamefont {L.-C.}\
  \bibnamefont {Peng}}, \bibinfo {author} {\bibfnamefont {Y.-H.}\ \bibnamefont
  {Luo}}, \bibinfo {author} {\bibfnamefont {D.}~\bibnamefont {Wu}}, \bibinfo
  {author} {\bibfnamefont {S.-Q.}\ \bibnamefont {Gong}}, \bibinfo {author}
  {\bibfnamefont {H.}~\bibnamefont {Su}}, \bibinfo {author} {\bibfnamefont
  {Y.}~\bibnamefont {Hu}}, \bibinfo {author} {\bibfnamefont {P.}~\bibnamefont
  {Hu}}, \bibinfo {author} {\bibfnamefont {X.-Y.}\ \bibnamefont {Yang}},
  \bibinfo {author} {\bibfnamefont {W.-J.}\ \bibnamefont {Zhang}}, \bibinfo
  {author} {\bibfnamefont {H.}~\bibnamefont {Li}}, \bibinfo {author}
  {\bibfnamefont {Y.}~\bibnamefont {Li}}, \bibinfo {author} {\bibfnamefont
  {X.}~\bibnamefont {Jiang}}, \bibinfo {author} {\bibfnamefont
  {L.}~\bibnamefont {Gan}}, \bibinfo {author} {\bibfnamefont {G.}~\bibnamefont
  {Yang}}, \bibinfo {author} {\bibfnamefont {L.}~\bibnamefont {You}}, \bibinfo
  {author} {\bibfnamefont {Z.}~\bibnamefont {Wang}}, \bibinfo {author}
  {\bibfnamefont {L.}~\bibnamefont {Li}}, \bibinfo {author} {\bibfnamefont
  {N.-L.}\ \bibnamefont {Liu}}, \bibinfo {author} {\bibfnamefont
  {J.J.}~\bibnamefont {Renema}}, \bibinfo {author} {\bibfnamefont {C.-Y.}\
  \bibnamefont {Lu}},\ and\ \bibinfo {author} {\bibfnamefont {J.-W.}\
  \bibnamefont {Pan}},\ }\bibfield  {title} {\bibinfo {title}
  {Phase-{{Programmable Gaussian Boson Sampling Using Stimulated Squeezed
  Light}}},\ }\href {https://doi.org/10.1103/PhysRevLett.127.180502} {\bibfield
   {journal} {\bibinfo  {journal} {Physical Review Letters}\ }\textbf {\bibinfo
  {volume} {127}},\ \bibinfo {pages} {180502} (\bibinfo {year}
  {2021})}\BibitemShut {NoStop}%
\bibitem [{\citenamefont {Deng}\ \emph {et~al.}(2023)\citenamefont {Deng},
  \citenamefont {Gong}, \citenamefont {Gu}, \citenamefont {Zhang},
  \citenamefont {Liu}, \citenamefont {Su}, \citenamefont {Tang}, \citenamefont
  {Xu}, \citenamefont {Jia}, \citenamefont {Chen}, \citenamefont {Zhong},
  \citenamefont {Wang}, \citenamefont {Yan}, \citenamefont {Hu}, \citenamefont
  {Huang}, \citenamefont {Zhang}, \citenamefont {Li}, \citenamefont {Jiang},
  \citenamefont {You}, \citenamefont {Wang}, \citenamefont {Li}, \citenamefont
  {Liu}, \citenamefont {Lu},\ and\ \citenamefont {Pan}}]{deng_solving_2023}%
  \BibitemOpen
  \bibfield  {author} {\bibinfo {author} {\bibfnamefont {Y.-H.}\ \bibnamefont
  {Deng}}, \bibinfo {author} {\bibfnamefont {S.-Q.}\ \bibnamefont {Gong}},
  \bibinfo {author} {\bibfnamefont {Y.-C.}\ \bibnamefont {Gu}}, \bibinfo
  {author} {\bibfnamefont {Z.-J.}\ \bibnamefont {Zhang}}, \bibinfo {author}
  {\bibfnamefont {H.-L.}\ \bibnamefont {Liu}}, \bibinfo {author} {\bibfnamefont
  {H.}~\bibnamefont {Su}}, \bibinfo {author} {\bibfnamefont {H.-Y.}\
  \bibnamefont {Tang}}, \bibinfo {author} {\bibfnamefont {J.-M.}\ \bibnamefont
  {Xu}}, \bibinfo {author} {\bibfnamefont {M.-H.}\ \bibnamefont {Jia}},
  \bibinfo {author} {\bibfnamefont {M.-C.}\ \bibnamefont {Chen}}, \bibinfo
  {author} {\bibfnamefont {H.-S.}\ \bibnamefont {Zhong}}, \bibinfo {author}
  {\bibfnamefont {H.}~\bibnamefont {Wang}}, \bibinfo {author} {\bibfnamefont
  {J.}~\bibnamefont {Yan}}, \bibinfo {author} {\bibfnamefont {Y.}~\bibnamefont
  {Hu}}, \bibinfo {author} {\bibfnamefont {J.}~\bibnamefont {Huang}}, \bibinfo
  {author} {\bibfnamefont {W.-J.}\ \bibnamefont {Zhang}}, \bibinfo {author}
  {\bibfnamefont {H.}~\bibnamefont {Li}}, \bibinfo {author} {\bibfnamefont
  {X.}~\bibnamefont {Jiang}}, \bibinfo {author} {\bibfnamefont
  {L.}~\bibnamefont {You}}, \bibinfo {author} {\bibfnamefont {Z.}~\bibnamefont
  {Wang}}, \bibinfo {author} {\bibfnamefont {L.}~\bibnamefont {Li}}, \bibinfo
  {author} {\bibfnamefont {N.-L.}\ \bibnamefont {Liu}}, \bibinfo {author}
  {\bibfnamefont {C.-Y.}\ \bibnamefont {Lu}},\ and\ \bibinfo {author}
  {\bibfnamefont {J.-W.}\ \bibnamefont {Pan}},\ }\bibfield  {title} {\bibinfo
  {title} {Solving {{Graph Problems Using Gaussian Boson Sampling}}},\ }\href
  {https://doi.org/10.1103/PhysRevLett.130.190601} {\bibfield  {journal}
  {\bibinfo  {journal} {Physical Review Letters}\ }\textbf {\bibinfo {volume}
  {130}},\ \bibinfo {pages} {190601} (\bibinfo {year} {2023})}\BibitemShut
  {NoStop}%
\bibitem [{\citenamefont {Deng}\ \emph {et~al.}()\citenamefont {Deng},
  \citenamefont {Gu}, \citenamefont {Liu}, \citenamefont {Gong}, \citenamefont
  {Su}, \citenamefont {Zhang}, \citenamefont {Tang}, \citenamefont {Jia},
  \citenamefont {Xu}, \citenamefont {Chen}, \citenamefont {Qin}, \citenamefont
  {Peng}, \citenamefont {Yan}, \citenamefont {Hu}, \citenamefont {Huang},
  \citenamefont {Li}, \citenamefont {Li}, \citenamefont {Chen}, \citenamefont
  {Jiang}, \citenamefont {Gan}, \citenamefont {Yang}, \citenamefont {You},
  \citenamefont {Li}, \citenamefont {Zhong}, \citenamefont {Wang},
  \citenamefont {Liu}, \citenamefont {Renema}, \citenamefont {Lu},\ and\
  \citenamefont {Pan}}]{deng_gaussian_2023}%
  \BibitemOpen
  \bibfield  {author} {\bibinfo {author} {\bibfnamefont {Y.-H.}\ \bibnamefont
  {Deng}}, \bibinfo {author} {\bibfnamefont {Y.-C.}\ \bibnamefont {Gu}},
  \bibinfo {author} {\bibfnamefont {H.-L.}\ \bibnamefont {Liu}}, \bibinfo
  {author} {\bibfnamefont {S.-Q.}\ \bibnamefont {Gong}}, \bibinfo {author}
  {\bibfnamefont {H.}~\bibnamefont {Su}}, \bibinfo {author} {\bibfnamefont
  {Z.-J.}\ \bibnamefont {Zhang}}, \bibinfo {author} {\bibfnamefont {H.-Y.}\
  \bibnamefont {Tang}}, \bibinfo {author} {\bibfnamefont {M.-H.}\ \bibnamefont
  {Jia}}, \bibinfo {author} {\bibfnamefont {J.-M.}\ \bibnamefont {Xu}},
  \bibinfo {author} {\bibfnamefont {M.-C.}\ \bibnamefont {Chen}}, \bibinfo
  {author} {\bibfnamefont {J.}~\bibnamefont {Qin}}, \bibinfo {author}
  {\bibfnamefont {L.-C.}\ \bibnamefont {Peng}}, \bibinfo {author}
  {\bibfnamefont {J.}~\bibnamefont {Yan}}, \bibinfo {author} {\bibfnamefont
  {Y.}~\bibnamefont {Hu}}, \bibinfo {author} {\bibfnamefont {J.}~\bibnamefont
  {Huang}}, \bibinfo {author} {\bibfnamefont {H.}~\bibnamefont {Li}}, \bibinfo
  {author} {\bibfnamefont {Y.}~\bibnamefont {Li}}, \bibinfo {author}
  {\bibfnamefont {Y.}~\bibnamefont {Chen}}, \bibinfo {author} {\bibfnamefont
  {X.}~\bibnamefont {Jiang}}, \bibinfo {author} {\bibfnamefont
  {L.}~\bibnamefont {Gan}}, \bibinfo {author} {\bibfnamefont {G.}~\bibnamefont
  {Yang}}, \bibinfo {author} {\bibfnamefont {L.}~\bibnamefont {You}}, \bibinfo
  {author} {\bibfnamefont {L.}~\bibnamefont {Li}}, \bibinfo {author}
  {\bibfnamefont {H.-S.}\ \bibnamefont {Zhong}}, \bibinfo {author}
  {\bibfnamefont {H.}~\bibnamefont {Wang}}, \bibinfo {author} {\bibfnamefont
  {N.-L.}\ \bibnamefont {Liu}}, \bibinfo {author} {\bibfnamefont {J.~J.}\
  \bibnamefont {Renema}}, \bibinfo {author} {\bibfnamefont {C.-Y.}\
  \bibnamefont {Lu}},\ and\ \bibinfo {author} {\bibfnamefont {J.-W.}\
  \bibnamefont {Pan}},\ }\bibfield  {title} {\bibinfo {title} {Gaussian {{Boson
  Sampling}} with {{Pseudo-Photon-Number Resolving Detectors}} and {{Quantum
  Computational Advantage}}},\ }\href@noop {} {\bibinfo  {journal}
  {arXiv:2304.12240}\ }\BibitemShut {NoStop}%
\bibitem [{\citenamefont {Oh}\ \emph {et~al.}({\natexlab{a}})\citenamefont
  {Oh}, \citenamefont {Liu}, \citenamefont {Alexeev}, \citenamefont
  {Fefferman},\ and\ \citenamefont {Jiang}}]{oh_classical_2023-1}%
  \BibitemOpen
\bibfield  {journal} {  }\bibfield  {author} {\bibinfo {author} {\bibfnamefont
  {C.}~\bibnamefont {Oh}}, \bibinfo {author} {\bibfnamefont {M.}~\bibnamefont
  {Liu}}, \bibinfo {author} {\bibfnamefont {Y.}~\bibnamefont {Alexeev}},
  \bibinfo {author} {\bibfnamefont {B.}~\bibnamefont {Fefferman}},\ and\
  \bibinfo {author} {\bibfnamefont {L.}~\bibnamefont {Jiang}},\ }\bibfield
  {title} {\bibinfo {title} {Classical algorithm for simulating experimental
  {{Gaussian}} boson sampling},\ }\href@noop {} {\bibfield  {journal} {\bibinfo
   {journal} {arXiv:2306.03709}\ } ({\natexlab{a}})}\BibitemShut {NoStop}%
\bibitem [{\citenamefont {Arrazola}\ and\ \citenamefont
  {Bromley}(2018)}]{arrazola_using_2018}%
  \BibitemOpen
  \bibfield  {author} {\bibinfo {author} {\bibfnamefont {J.~M.}\ \bibnamefont
  {Arrazola}}\ and\ \bibinfo {author} {\bibfnamefont {T.~R.}\ \bibnamefont
  {Bromley}},\ }\bibfield  {title} {\bibinfo {title} {Using {{Gaussian Boson
  Sampling}} to {{Find Dense Subgraphs}}},\ }\href
  {https://doi.org/10.1103/PhysRevLett.121.030503} {\bibfield  {journal}
  {\bibinfo  {journal} {Physical Review Letters}\ }\textbf {\bibinfo {volume}
  {121}},\ \bibinfo {pages} {030503} (\bibinfo {year} {2018})}\BibitemShut
  {NoStop}%
\bibitem [{\citenamefont {Arrazola}\ \emph {et~al.}(2018)\citenamefont
  {Arrazola}, \citenamefont {Bromley},\ and\ \citenamefont
  {Rebentrost}}]{arrazola_quantum_2018}%
  \BibitemOpen
  \bibfield  {author} {\bibinfo {author} {\bibfnamefont {J.~M.}\ \bibnamefont
  {Arrazola}}, \bibinfo {author} {\bibfnamefont {T.~R.}\ \bibnamefont
  {Bromley}},\ and\ \bibinfo {author} {\bibfnamefont {P.}~\bibnamefont
  {Rebentrost}},\ }\bibfield  {title} {\bibinfo {title} {Quantum approximate
  optimization with {{Gaussian}} boson sampling},\ }\href
  {https://doi.org/10.1103/PhysRevA.98.012322} {\bibfield  {journal} {\bibinfo
  {journal} {Physical Review A}\ }\textbf {\bibinfo {volume} {98}},\ \bibinfo
  {pages} {012322} (\bibinfo {year} {2018})}\BibitemShut {NoStop}%
\bibitem [{\citenamefont {Bromley}\ \emph {et~al.}(2020)\citenamefont
  {Bromley}, \citenamefont {Arrazola}, \citenamefont {Jahangiri}, \citenamefont
  {Izaac}, \citenamefont {Quesada}, \citenamefont {Gran}, \citenamefont
  {Schuld}, \citenamefont {Swinarton}, \citenamefont {Zabaneh},\ and\
  \citenamefont {Killoran}}]{bromley_applications_2020}%
  \BibitemOpen
  \bibfield  {author} {\bibinfo {author} {\bibfnamefont {T.~R.}\ \bibnamefont
  {Bromley}}, \bibinfo {author} {\bibfnamefont {J.~M.}\ \bibnamefont
  {Arrazola}}, \bibinfo {author} {\bibfnamefont {S.}~\bibnamefont {Jahangiri}},
  \bibinfo {author} {\bibfnamefont {J.}~\bibnamefont {Izaac}}, \bibinfo
  {author} {\bibfnamefont {N.}~\bibnamefont {Quesada}}, \bibinfo {author}
  {\bibfnamefont {A.~D.}\ \bibnamefont {Gran}}, \bibinfo {author}
  {\bibfnamefont {M.}~\bibnamefont {Schuld}}, \bibinfo {author} {\bibfnamefont
  {J.}~\bibnamefont {Swinarton}}, \bibinfo {author} {\bibfnamefont
  {Z.}~\bibnamefont {Zabaneh}},\ and\ \bibinfo {author} {\bibfnamefont
  {N.}~\bibnamefont {Killoran}},\ }\bibfield  {title} {\bibinfo {title}
  {Applications of near-term photonic quantum computers: Software and
  algorithms},\ }\href {https://doi.org/10.1088/2058-9565/ab8504} {\bibfield
  {journal} {\bibinfo  {journal} {Quantum Science and Technology}\ }\textbf
  {\bibinfo {volume} {5}},\ \bibinfo {pages} {034010} (\bibinfo {year}
  {2020})}\BibitemShut {NoStop}%
\bibitem [{\citenamefont {Banchi}\ \emph {et~al.}(2020)\citenamefont {Banchi},
  \citenamefont {Fingerhuth}, \citenamefont {Babej}, \citenamefont {Ing},\ and\
  \citenamefont {Arrazola}}]{banchi_molecular_2020}%
  \BibitemOpen
  \bibfield  {author} {\bibinfo {author} {\bibfnamefont {L.}~\bibnamefont
  {Banchi}}, \bibinfo {author} {\bibfnamefont {M.}~\bibnamefont {Fingerhuth}},
  \bibinfo {author} {\bibfnamefont {T.}~\bibnamefont {Babej}}, \bibinfo
  {author} {\bibfnamefont {C.}~\bibnamefont {Ing}},\ and\ \bibinfo {author}
  {\bibfnamefont {J.~M.}\ \bibnamefont {Arrazola}},\ }\bibfield  {title}
  {\bibinfo {title} {Molecular docking with {{Gaussian Boson Sampling}}},\
  }\href {https://doi.org/10.1126/sciadv.aax1950} {\bibfield  {journal}
  {\bibinfo  {journal} {Science Advances}\ }\textbf {\bibinfo {volume} {6}},\
  \bibinfo {pages} {eaax1950} (\bibinfo {year} {2020})}\BibitemShut {NoStop}%
\bibitem [{\citenamefont {Anteneh}\ and\ \citenamefont
  {Pfister}(2023)}]{anteneh_sample_2023}%
  \BibitemOpen
  \bibfield  {author} {\bibinfo {author} {\bibfnamefont {A.}~\bibnamefont
  {Anteneh}}\ and\ \bibinfo {author} {\bibfnamefont {O.}~\bibnamefont
  {Pfister}},\ }\href@noop {} {\bibinfo {title} {Sample efficient graph
  classification using binary {{Gaussian}} boson sampling}} (\bibinfo {year}
  {2023})\BibitemShut {NoStop}%
\bibitem [{\citenamefont {Schuld}\ \emph {et~al.}(2020)\citenamefont {Schuld},
  \citenamefont {Br{\'a}dler}, \citenamefont {Israel}, \citenamefont {Su},\
  and\ \citenamefont {Gupt}}]{schuld_measuring_2020}%
  \BibitemOpen
  \bibfield  {author} {\bibinfo {author} {\bibfnamefont {M.}~\bibnamefont
  {Schuld}}, \bibinfo {author} {\bibfnamefont {K.}~\bibnamefont {Br{\'a}dler}},
  \bibinfo {author} {\bibfnamefont {R.}~\bibnamefont {Israel}}, \bibinfo
  {author} {\bibfnamefont {D.}~\bibnamefont {Su}},\ and\ \bibinfo {author}
  {\bibfnamefont {B.}~\bibnamefont {Gupt}},\ }\bibfield  {title} {\bibinfo
  {title} {Measuring the similarity of graphs with a {{Gaussian}} boson
  sampler},\ }\href {https://doi.org/10.1103/PhysRevA.101.032314} {\bibfield
  {journal} {\bibinfo  {journal} {Physical Review A}\ }\textbf {\bibinfo
  {volume} {101}},\ \bibinfo {pages} {032314} (\bibinfo {year}
  {2020})}\BibitemShut {NoStop}%
\bibitem [{\citenamefont {Quesada}\ and\ \citenamefont
  {Arrazola}(2020)}]{quesada_exact_2020}%
  \BibitemOpen
  \bibfield  {author} {\bibinfo {author} {\bibfnamefont {N.}~\bibnamefont
  {Quesada}}\ and\ \bibinfo {author} {\bibfnamefont {J.~M.}\ \bibnamefont
  {Arrazola}},\ }\bibfield  {title} {\bibinfo {title} {Exact simulation of
  {{Gaussian}} boson sampling in polynomial space and exponential time},\
  }\href {https://doi.org/10.1103/PhysRevResearch.2.023005} {\bibfield
  {journal} {\bibinfo  {journal} {Physical Review Research}\ }\textbf {\bibinfo
  {volume} {2}},\ \bibinfo {pages} {023005} (\bibinfo {year}
  {2020})}\BibitemShut {NoStop}%
\bibitem [{\citenamefont {Barvinok}(1999)}]{barvinok_polynomial_1999}%
  \BibitemOpen
  \bibfield  {author} {\bibinfo {author} {\bibfnamefont {A.}~\bibnamefont
  {Barvinok}},\ }\bibfield  {title} {\bibinfo {title} {Polynomial {{Time
  Algorithms}} to {{Approximate Permanents}} and {{Mixed Discriminants Within}}
  a {{Simply Exponential Factor}}},\ }\href
  {https://doi.org/10.1002/(SICI)1098-2418(1999010)14:1<29::AID-RSA2>3.0.CO;2-X}
  {\bibfield  {journal} {\bibinfo  {journal} {Random Structures and
  Algorithms}\ }\textbf {\bibinfo {volume} {14}},\ \bibinfo {pages} {29}
  (\bibinfo {year} {1999})}\BibitemShut {NoStop}%
\bibitem [{\citenamefont {Godsil}\ and\ \citenamefont
  {Gutman}(1978)}]{godsil1978matching}%
  \BibitemOpen
  \bibfield  {author} {\bibinfo {author} {\bibfnamefont {C.~D.}\ \bibnamefont
  {Godsil}}\ and\ \bibinfo {author} {\bibfnamefont {I.}~\bibnamefont
  {Gutman}},\ }\href@noop {} {\emph {\bibinfo {title} {On the matching
  polynomial of a graph}}}\ (\bibinfo  {publisher} {University of Melbourne},\
  \bibinfo {year} {1978})\BibitemShut {NoStop}%
\bibitem [{\citenamefont {Karmarkar}\ \emph {et~al.}(1993)\citenamefont
  {Karmarkar}, \citenamefont {Karp}, \citenamefont {Lipton}, \citenamefont
  {Lov{\'a}sz},\ and\ \citenamefont {Luby}}]{karmarkar_monte-carlo_1993}%
  \BibitemOpen
  \bibfield  {author} {\bibinfo {author} {\bibfnamefont {N.}~\bibnamefont
  {Karmarkar}}, \bibinfo {author} {\bibfnamefont {R.}~\bibnamefont {Karp}},
  \bibinfo {author} {\bibfnamefont {R.}~\bibnamefont {Lipton}}, \bibinfo
  {author} {\bibfnamefont {L.}~\bibnamefont {Lov{\'a}sz}},\ and\ \bibinfo
  {author} {\bibfnamefont {M.}~\bibnamefont {Luby}},\ }\bibfield  {title}
  {\bibinfo {title} {A {{Monte-Carlo Algorithm}} for {{Estimating}} the
  {{Permanent}}},\ }\href {https://doi.org/10.1137/0222021} {\bibfield
  {journal} {\bibinfo  {journal} {SIAM Journal on Computing}\ }\textbf
  {\bibinfo {volume} {22}},\ \bibinfo {pages} {284} (\bibinfo {year}
  {1993})}\BibitemShut {NoStop}%
\bibitem [{\citenamefont {Lov{\'a}sz}\ and\ \citenamefont
  {Plummer}(2009)}]{lovasz_matching_2009}%
  \BibitemOpen
  \bibfield  {author} {\bibinfo {author} {\bibfnamefont {L.}~\bibnamefont
  {Lov{\'a}sz}}\ and\ \bibinfo {author} {\bibfnamefont {M.~D.}\ \bibnamefont
  {Plummer}},\ }\href@noop {} {\emph {\bibinfo {title} {Matching Theory}}},\
  \bibinfo {edition} {repr. with corr}\ ed.\ (\bibinfo  {publisher} {{AMS
  Chelsea Publ}},\ \bibinfo {address} {{Providence, RI}},\ \bibinfo {year}
  {2009})\BibitemShut {NoStop}%
\bibitem [{\citenamefont {Joyal}(1981)}]{joyal_theorie_1981}%
  \BibitemOpen
  \bibfield  {author} {\bibinfo {author} {\bibfnamefont {A.}~\bibnamefont
  {Joyal}},\ }\bibfield  {title} {\bibinfo {title} {{Une th\'eorie combinatoire
  des s\'eries formelles}},\ }\href
  {https://doi.org/10.1016/0001-8708(81)90052-9} {\bibfield  {journal}
  {\bibinfo  {journal} {Advances in Mathematics}\ }\textbf {\bibinfo {volume}
  {42}},\ \bibinfo {pages} {1} (\bibinfo {year} {1981})}\BibitemShut {NoStop}%
\bibitem [{\citenamefont {Bergeron}\ \emph {et~al.}(1997)\citenamefont
  {Bergeron}, \citenamefont {Labelle},\ and\ \citenamefont
  {Leroux}}]{bergeron_combinatorial_1997}%
  \BibitemOpen
  \bibfield  {author} {\bibinfo {author} {\bibfnamefont {F.}~\bibnamefont
  {Bergeron}}, \bibinfo {author} {\bibfnamefont {G.}~\bibnamefont {Labelle}},\
  and\ \bibinfo {author} {\bibfnamefont {P.}~\bibnamefont {Leroux}},\ }\href
  {https://doi.org/10.1017/CBO9781107325913} {\emph {\bibinfo {title}
  {Combinatorial {{Species}} and {{Tree-like Structures}}}}},\ \bibinfo
  {edition} {1st}\ ed.\ (\bibinfo  {publisher} {{Cambridge University Press}},\
  \bibinfo {year} {1997})\BibitemShut {NoStop}%
\bibitem [{\citenamefont {Wilf}(2005)}]{wilf2005generatingfunctionology}%
  \BibitemOpen
  \bibfield  {author} {\bibinfo {author} {\bibfnamefont {H.~S.}\ \bibnamefont
  {Wilf}},\ }\href@noop {} {\emph {\bibinfo {title}
  {generatingfunctionology}}}\ (\bibinfo  {publisher} {CRC press},\ \bibinfo
  {year} {2005})\BibitemShut {NoStop}%
\bibitem [{\citenamefont {Flajolet}\ and\ \citenamefont
  {Sedgewick}(2009)}]{flajolet2009analytic}%
  \BibitemOpen
  \bibfield  {author} {\bibinfo {author} {\bibfnamefont {P.}~\bibnamefont
  {Flajolet}}\ and\ \bibinfo {author} {\bibfnamefont {R.}~\bibnamefont
  {Sedgewick}},\ }\href@noop {} {\emph {\bibinfo {title} {Analytic
  combinatorics}}}\ (\bibinfo  {publisher} {cambridge University press},\
  \bibinfo {year} {2009})\BibitemShut {NoStop}%
\bibitem [{\citenamefont {Renema}\ \emph {et~al.}(2018)\citenamefont {Renema},
  \citenamefont {Menssen}, \citenamefont {Clements}, \citenamefont {Triginer},
  \citenamefont {Kolthammer},\ and\ \citenamefont
  {Walmsley}}]{renema_efficient_2018-1}%
  \BibitemOpen
  \bibfield  {author} {\bibinfo {author} {\bibfnamefont {J.~J.}\ \bibnamefont
  {Renema}}, \bibinfo {author} {\bibfnamefont {A.}~\bibnamefont {Menssen}},
  \bibinfo {author} {\bibfnamefont {W.~R.}\ \bibnamefont {Clements}}, \bibinfo
  {author} {\bibfnamefont {G.}~\bibnamefont {Triginer}}, \bibinfo {author}
  {\bibfnamefont {W.~S.}\ \bibnamefont {Kolthammer}},\ and\ \bibinfo {author}
  {\bibfnamefont {I.~A.}\ \bibnamefont {Walmsley}},\ }\bibfield  {title}
  {\bibinfo {title} {Efficient {{Classical Algorithm}} for {{Boson Sampling}}
  with {{Partially Distinguishable Photons}}},\ }\href
  {https://doi.org/10.1103/PhysRevLett.120.220502} {\bibfield  {journal}
  {\bibinfo  {journal} {Physical Review Letters}\ }\textbf {\bibinfo {volume}
  {120}},\ \bibinfo {pages} {220502} (\bibinfo {year} {2018})}\BibitemShut
  {NoStop}%
\bibitem [{\citenamefont {Renema}\ \emph {et~al.}()\citenamefont {Renema},
  \citenamefont {Shchesnovich},\ and\ \citenamefont
  {{Garcia-Patron}}}]{renema_classical_2019}%
  \BibitemOpen
  \bibfield  {author} {\bibinfo {author} {\bibfnamefont {J.}~\bibnamefont
  {Renema}}, \bibinfo {author} {\bibfnamefont {V.}~\bibnamefont
  {Shchesnovich}},\ and\ \bibinfo {author} {\bibfnamefont {R.}~\bibnamefont
  {{Garcia-Patron}}},\ }\bibfield  {title} {\bibinfo {title} {Classical
  simulability of noisy boson sampling},\ }\href@noop {} {\bibinfo  {journal}
  {arXiv:1809.01953}\ }\BibitemShut {NoStop}%
\bibitem [{\citenamefont {Villalonga}\ \emph {et~al.}()\citenamefont
  {Villalonga}, \citenamefont {Niu}, \citenamefont {Li}, \citenamefont {Neven},
  \citenamefont {Platt}, \citenamefont {Smelyanskiy},\ and\ \citenamefont
  {Boixo}}]{villalonga_efficient_2022}%
  \BibitemOpen
\bibfield  {journal} {  }\bibfield  {author} {\bibinfo {author} {\bibfnamefont
  {B.}~\bibnamefont {Villalonga}}, \bibinfo {author} {\bibfnamefont {M.~Y.}\
  \bibnamefont {Niu}}, \bibinfo {author} {\bibfnamefont {L.}~\bibnamefont
  {Li}}, \bibinfo {author} {\bibfnamefont {H.}~\bibnamefont {Neven}}, \bibinfo
  {author} {\bibfnamefont {J.~C.}\ \bibnamefont {Platt}}, \bibinfo {author}
  {\bibfnamefont {V.~N.}\ \bibnamefont {Smelyanskiy}},\ and\ \bibinfo {author}
  {\bibfnamefont {S.}~\bibnamefont {Boixo}},\ }\bibfield  {title} {\bibinfo
  {title} {Efficient approximation of experimental {{Gaussian}} boson
  sampling},\ }\href@noop {} {\bibinfo  {journal} {arXiv:2109.11525}\
  }\BibitemShut {NoStop}%
\bibitem [{\citenamefont {Popova}\ and\ \citenamefont
  {Rubtsov}()}]{popova_cracking_2022}%
  \BibitemOpen
\bibfield  {journal} {  }\bibfield  {author} {\bibinfo {author} {\bibfnamefont
  {A.~S.}\ \bibnamefont {Popova}}\ and\ \bibinfo {author} {\bibfnamefont
  {A.~N.}\ \bibnamefont {Rubtsov}},\ }\bibfield  {title} {\bibinfo {title}
  {Cracking the {{Quantum Advantage}} threshold for {{Gaussian Boson
  Sampling}}},\ }\href@noop {} {\bibinfo  {journal} {arXiv:2106.01445}\
  }\BibitemShut {NoStop}%
\bibitem [{\citenamefont {Solomons}\ \emph {et~al.}()\citenamefont {Solomons},
  \citenamefont {Thomas},\ and\ \citenamefont
  {McCutcheon}}]{solomons_gaussian-boson-sampling-enhanced_2023}%
  \BibitemOpen
\bibfield  {journal} {  }\bibfield  {author} {\bibinfo {author} {\bibfnamefont
  {N.~R.}\ \bibnamefont {Solomons}}, \bibinfo {author} {\bibfnamefont {O.~F.}\
  \bibnamefont {Thomas}},\ and\ \bibinfo {author} {\bibfnamefont {D.~P.~S.}\
  \bibnamefont {McCutcheon}},\ }\bibfield  {title} {\bibinfo {title}
  {Gaussian-boson-sampling-enhanced dense subgraph finding shows limited
  advantage over efficient classical algorithms},\ }\href@noop {} {\bibinfo
  {journal} {arXiv:2301.13217}\ }\BibitemShut {NoStop}%
\bibitem [{\citenamefont {Oh}\ \emph {et~al.}({\natexlab{b}})\citenamefont
  {Oh}, \citenamefont {Fefferman}, \citenamefont {Jiang},\ and\ \citenamefont
  {Quesada}}]{oh_quantum-inspired_2023-1}%
  \BibitemOpen
\bibfield  {journal} {  }\bibfield  {author} {\bibinfo {author} {\bibfnamefont
  {C.}~\bibnamefont {Oh}}, \bibinfo {author} {\bibfnamefont {B.}~\bibnamefont
  {Fefferman}}, \bibinfo {author} {\bibfnamefont {L.}~\bibnamefont {Jiang}},\
  and\ \bibinfo {author} {\bibfnamefont {N.}~\bibnamefont {Quesada}},\
  }\bibfield  {title} {\bibinfo {title} {Quantum-inspired classical algorithm
  for graph problems by {{Gaussian}} boson sampling},\ }\href@noop {}
  {\bibfield  {journal} {\bibinfo  {journal} {arXiv:2302.00536}\ }
  ({\natexlab{b}})}\BibitemShut {NoStop}%
\bibitem [{\citenamefont {Gupt}\ \emph {et~al.}(2019)\citenamefont {Gupt},
  \citenamefont {Izaac},\ and\ \citenamefont {Quesada}}]{gupt_walrus_2019}%
  \BibitemOpen
  \bibfield  {author} {\bibinfo {author} {\bibfnamefont {B.}~\bibnamefont
  {Gupt}}, \bibinfo {author} {\bibfnamefont {J.}~\bibnamefont {Izaac}},\ and\
  \bibinfo {author} {\bibfnamefont {N.}~\bibnamefont {Quesada}},\ }\bibfield
  {title} {\bibinfo {title} {The {{Walrus}}: A library for the calculation of
  hafnians, {{Hermite}} polynomials and {{Gaussian}} boson sampling},\ }\href
  {https://doi.org/10.21105/joss.01705} {\bibfield  {journal} {\bibinfo
  {journal} {Journal of Open Source Software}\ }\textbf {\bibinfo {volume}
  {4}},\ \bibinfo {pages} {1705} (\bibinfo {year} {2019})}\BibitemShut
  {NoStop}%
\end{thebibliography}%

\end{document}